\documentclass[a4paper,11pt]{article}

\usepackage{amsmath}
\usepackage{amsthm}
\usepackage{authblk}
\usepackage{mathpazo}
\usepackage{palatino}
\usepackage{setspace}
\usepackage{url}
\usepackage[comma,authoryear]{natbib}
\usepackage{hyperref}
\usepackage{theoremref}
\usepackage[titletoc,title]{appendix}
\usepackage[table]{xcolor}
\usepackage{float}
\restylefloat{table}
\usepackage{amssymb}
\usepackage{caption}
\usepackage{subcaption}
\usepackage[shortlabels]{enumitem}

\definecolor{armygreen}{rgb}{0.29, 0.33, 0.13}
\definecolor{auburn}{rgb}{0.43, 0.21, 0.1}
\definecolor{burgundy}{rgb}{0.5, 0.0, 0.13}
\definecolor{medium red}{rgb}{.490,.298,.337}
\definecolor{dark red}{rgb}{.235,.141,.161}
\definecolor{dark green}{rgb}{0.0,0.5,0.0}

\hypersetup{
	colorlinks = true,
	linkcolor = {burgundy},
	urlcolor = {burgundy},
	citecolor = {burgundy}, 		
	linkbordercolor = {white},
}

\newtheorem{theorem}{Theorem}[section]
\newtheorem{proposition}{Proposition}[section]
\newtheorem{claim}{Claim}[section]

\newtheorem{corollary}{Corollary}[section]
\newtheorem{obs}{Observation}[section]

\theoremstyle{definition}
\newtheorem{definition}{Definition}[section]
\theoremstyle{definition}
\newtheorem{example}{Example}[section]
\theoremstyle{definition}
\newtheorem{remark}{Remark}[section]
\theoremstyle{definition}
\newtheorem{note}{Note}[section]

\newcommand*{\claimproofname}{Proof of Claim}
\newenvironment{claimproof}[1][\claimproofname]{\begin{proof}[#1]}{\end{proof}}

\usepackage{verbatim}
\usepackage{tikz}
\tikzset{
	treenode/.style = {shape=rectangle, rounded corners,
		draw, align=center,
		top color=white, bottom color=blue!20},
	root/.style  = {shape=rectangle, rounded corners,
		draw, align=center,
		top color=white, bottom color=blue!20},
	root1/.style  = {shape=rectangle, rounded corners,
		draw, align=center,
		top color=white, bottom color=red!40},
	root2/.style  = {shape=rectangle, rounded corners,
		draw, align=center,
		top color=white, bottom color=red!20},
	env/.style  = {shape=rectangle, rounded corners,
		draw, align=center,
		top color=white, bottom color=blue!20},
}

\title{Simple dominance of fixed priority top trading cycles\thanks{I am grateful to Oishee Banerjee for helpful suggestions. I am particularly indebted to Szilvia P\'{a}pai whose insightful comments led to a significant improvement of the contents of this paper. Thanks also to Clayton Thomas for his questions and comments.}}

\author{Pinaki Mandal\thanks{E-mail: \textit{pinaki.mandal@asu.edu}} }

\affil{Department of Economics, Arizona State University, USA}

\date{ }

\begin{document} 	
	\maketitle
	
	\begin{abstract}
		We study the implementation of fixed priority top trading cycles (FPTTC) rules via simply dominant mechanisms \citep{pycia2019theory} in the context of assignment problems, where agents are to be assigned at most one indivisible object and monetary transfers are not allowed. We consider both models -- with and without outside options, and characterize all simply dominant FPTTC rules in both models. We further introduce the notion of \textit{simple strategy-proofness} to resolve the issue with agents being concerned about having time-inconsistent preferences, and discuss its relation with simple dominance.
	\end{abstract}
	
	
	
	\newpage

	\maketitle

	\section{Introduction}\label{section introduction}

	We consider the well-known \textit{assignment problem} (also known as \textit{house allocation problem} or \textit{resource allocation problem}) where a set of heterogeneous indivisible objects are to be allocated among a group of agents so that each agent receives at most one object and monetary transfers are not allowed. Such problems arise when, for instance, the Government wants to assign houses to the citizens, or hospitals to doctors, or a manager wants to allocate offices to employees, or tasks to workers, or a professor wants to assign projects to students. Agents are asked to report their strict preferences over the objects and remaining unassigned (also known as the outside option), and a mechanism selects an allocation (of the objects among the agents) based on the agents' reports. Note that an object is acceptable to an agent if and only if she prefers that object to the outside option. 
	
	An important consideration while designing mechanisms is to implement desirable outcomes when the participating agents are strategic. The standard notion of \textit{strategy-proofness} requires truth-telling to be a weakly dominant strategy, that is, no agent can be strictly better off by misreporting her (true) preference.
	
	\textit{Fixed priority top trading cycles (FPTTC) mechanism} \citep{papai2000strategyproof} is a well-known strategy-proof way to assign objects in the absence of transfers.\footnote{\citet{papai2000strategyproof} uses the term ``fixed endowment hierarchical exchange rule'' to refer to an FPTTC mechanism.} An FPTTC mechanism works in steps. At each step, the objects (available at that step) are owned by certain agents who then trade their objects by forming \textit{top trading cycles (TTC)}.\footnote{TTC is due to David Gale and discussed in \citet{shapley1974cores}.} Ownership of the objects at the start of each step is determined by an exogenously given \textit{priority structure} -- a collection of strict priority rankings of the agents indexed by the objects.
	
	Apart from strategy-proofness, FPTTC mechanisms also satisfy \textit{group strategy-proofness}, \textit{Pareto efficiency}, and \textit{non-bossiness} \citep{satterthwaite1981strategy}. Group strategy-proofness ensures that no group of agents can jointly manipulate their reports so that all of them weakly benefit from this manipulation, while at least one agent in the group strictly benefits. Pareto efficiency ensures that there is no other way to allocate the objects so that each agent is weakly better off (and hence some agent is strictly better-off). Notice that no agent is assigned an unacceptable object in a Pareto efficient allocation. Non-bossiness says that an agent cannot change the assignment of another one without changing her own assignment. However, despite their appealing theoretical properties, the use of FPTTC mechanisms in practice is rare as participating agents find it difficult to understand them, particularly the fact that these mechanisms are strategy-proof (see, for example, \citet{pathak2017really}).\footnote{Similar phenomena are also observed in other settings, see \citet{chen2006school}, \citet{hassidim2017mechanism}, \citet{hassidim2018strategic}, \citet{rees2018suboptimal}, and \citet{shorrer2018obvious} for details.}
	
	The notion of \textit{obvious strategy-proofness (OSP)} \citep{li2017obviously} has emerged as a remedy by strengthening strategy-proofness in a way so that it becomes transparent to the participating agents that a mechanism is not manipulable. The concept of OSP is based on the notion of \textit{obvious dominance} in an \textit{extensive-form game}. A strategy $s_i$ of an agent $i$ in an extensive-form game is obviously dominant if, whenever agent $i$ is called to play, even the worst possible final outcome from following $s_i$ is at least as good as the best possible outcome from following any deviating strategy $s'_i$ of agent $i$, where the best and worst cases are determined by considering all possible strategies that could be played by $i$'s opponents in the future, keeping her own strategy fixed. A mechanism is \textit{OSP-implementable} if one can construct an extensive-form game that has an equilibrium in obviously dominant strategies.
	
	While OSP relaxes the assumption that the participating agents fully comprehend how the strategies of opponents will affect outcomes, it still presumes that they understand how their own future actions will affect outcomes. In other words, when checking obvious dominance, the worst possible final outcome and the best possible final outcome are taken only over opponents' strategies $s_{-i}$, fixing the agent's own strategy $s_i$. \citet{pycia2019theory} argue that in this case the agents might be concerned about having time-inconsistent strategies or making a mistake while performing backward induction over their own future actions. As a remedy, they introduce a natural strengthening of OSP called \textit{strongly obvious strategy-proofness (SOSP)} by relaxing the assumption that the agents understand how their own future actions will affect outcomes. A strategic plan is strongly obviously dominant if, whenever an agent is called to play, even the worst possible final outcome from the prescribed action is at least as good as the best possible outcome from any other action, where what is possible may depend on all future actions, including actions by the agent's future-self. Thus, strongly obviously dominant strategies are those that are weakly better than all alternative actions even if the agent is concerned that she might have time-inconsistent strategies.\footnote{This verbal description of SOSP is adapted from \citet{pycia2019theory}.}

	\subsection{Our motivation and contribution}

	\citet{mandal2022obviously,mandal2022outside} study the OSP-implementability of FPTTC mechanisms. They introduce the notion of \textit{dual ownership} and show that it is both necessary and sufficient condition for an FPTTC mechanism to be OSP-implementable. An FPTTC mechanism satisfies dual ownership if, for each preference profile and each step of the FPTTC mechanism at that preference profile, there are at most two agents who own all the objects available at that step. Although dual ownership is an intuitive property (and thereby, is quite helpful for explaining it to the participating agents), it is not so convenient for the designer to check whether a given FPTTC mechanism satisfies this property or not. This is because, technically, one needs to check at \textit{every} preference profile and \textit{every} step of the FPTTC mechanism at that preference profile, whether at most two agents are owning all the available objects at that step or not. This motivates one natural question: \textit{Is there any equivalent property to dual ownership, which is easier for the designer to check?}
	
	To tackle this question, we present two conditions for FPTTC mechanisms, namely \textit{acyclicity} and \textit{strong acyclicity} \citep{troyan2019obviously}, and show that they are equivalent to dual ownership in models without and with outside options, respectively (Theorem \ref{theorem FPTTC is dual iff acyclic} and Theorem \ref{theorem dual strong acyclicity equivalent}).\footnote{\citet{troyan2019obviously} uses the term ``TTC mechanism'' to refer to an FPTTC mechanism, and the term ``weak acyclicity'' to refer to strong acyclicity.}$^{,}$\footnote{In a model without outside options, every object is acceptable to every agent, whereas each object need not be acceptable to an agent in a model with outside options.} Both acyclicity and strong acyclicity are technical properties, which, as the names suggest, ensure that certain types of cycles are \textit{not} present in the associated priority structure of an FPTTC mechanism. The advantage of checking these properties for an FPTTC mechanism is that they only involve the associated priority structure, and not anything about the state of the FPTTC mechanism at different steps at different preference profiles. Since dual ownership property is intuitive but not convenient for the designer to check, whereas (strong) acyclicity property is technical but easier to check, these two equivalent properties, in a sense, complement each other. It is worth noting that both acyclicity and strong acyclicity are weaker than \textit{Ergin-acyclicity} \citep{ergin2002efficient}.\footnote{\citet{ergin2002efficient} introduces this notion in a many-to-one matching model.}$^{,}$\footnote{\citet{ehlers2002strategy} characterize strategy-proof, Pareto efficient, and population-monotonic mechanisms as FPTTC mechanisms satisfying Ergin-acyclicity (they use the term ``restricted endowment inheritance rules'' to refer to such FPTTC mechanisms). Later, \citet{ehlers2004resource} show that these mechanisms are also the only mechanisms satisfying Pareto efficiency, independence of irrelevant objects, and resource-monotonicity when there are more objects than agents (they use the term ``mixed dictator-pairwise-exchange rules'' to refer to such FPTTC mechanisms).}
	
	Apart from FPTTC mechanisms, another well-known class of priority-based mechanisms is \textit{agent-proposing deferred acceptance (APDA) mechanisms} \citep{gale1962college}. \citet{ashlagi2018stable} study the OSP-implementability of these mechanisms, and show that Ergin-acyclicity is a sufficient condition for an APDA mechanism to be OSP-implementable. We derive \citet{ashlagi2018stable}'s result as a corollary of our results (Corollary \ref{corollary ashlagi}). We further show that strong acyclicity is a necessary condition for an APDA mechanism to be OSP-implementable (Proposition \ref{proposition OSP is strong}).
	
	Next, we characterize the structure of SOSP-implementable FPTTC mechanisms. We introduce the notion of \textit{weak serial dictatorships} for this purpose. An FPTTC mechanism is a weak serial dictatorship if, for any preference profile and any step of the FPTTC mechanism at that preference profile, if there are more than two objects available at that step, then there is exactly one agent who owns all those objects. We characterize all SOSP-implementable FPTTC mechanisms in a model without outside options as weak serial dictatorships (Theorem \ref{theorem SOSP FPTTC restricted}). As a corollary of this result, we obtain that in a model without outside options and with more objects than agents, the class of SOSP-implementable FPTTC mechanisms is characterized by the class of \textit{serial dictatorships} \citep{satterthwaite1981strategy} (Corollary \ref{corollary more objects}). In a serial dictatorship, the agents choose their top choices among the ``remaining'' objects and the outside option, according to an ordering of the agents. We further show that the serial dictatorships are the only SOSP-implementable FPTTC mechanisms in a model with outside options (Theorem \ref{theorem SOSP FPTTC unrestricted}).
	
	Finally, we introduce the notion of \textit{simple strategy-proofness} which strengthens OSP-implementation to resolve the issue with agents being concerned about having time-inconsistent preferences. Recall that \citet{pycia2019theory} introduce SOSP-implementation for the same purpose. The concept of simple strategy-proofness is based on obvious dominance in a \textit{simple} extensive-form game. An extensive-form game is simple if every agent is called to play at most once. A mechanism is simple strategy-proof if one can construct a simple extensive-form game that has an equilibrium in obviously dominant strategies. We show that simple strategy-proofness is stronger than SOSP-implementability in general (Proposition \ref{proposition SSP implies SOSP}). We further show that the class of simply strategy-proof FPTTC mechanisms is the same as the class of SOSP-implementable FPTTC mechanisms in both models -- with and without outside options (Theorem \ref{theorem SSP SOSP FPTTC}).

	\subsection{Additional related literature}

	There is a rapidly growing body of work on OSP-implementability and various related notions of strategic simplicity. \citet{troyan2019obviously} shows that strong acyclicity is a sufficient condition for an FPTTC rule to be OSP-implementable when there is an equal number of agents and objects.\footnote{Theorem 1 in \citet{troyan2019obviously} says that in a model without outside options, strong acyclicity is both necessary and sufficient condition for an FPTTC rule to be OSP-implementable when there is an equal number of agents and objects. Later, \citet{mandal2022obviously} point out that while strong acyclicity is a sufficient condition for the same, it is not necessary (see Footnote 22 in \citet{mandal2022obviously} for details).} \citet{mandal2022obviously} characterize OSP-implementable, Pareto efficient, and non-bossy assignment rules as \textit{hierarchical exchange rules} \citep{papai2000strategyproof} satisfying dual ownership.\footnote{\citet{papai2000strategyproof} characterizes all strategy-proof, Pareto efficient, non-bossy, and reallocation-proof assignment rules as hierarchical exchange rules. Later, \citet{pycia2017incentive} introduce the notion of \textit{trading cycles rules} as a generalization of hierarchical exchange rules and show that an assignment rule is strategy-proof, Pareto efficient, and non-bossy if and only if it is a trading cycles rule.} \citet{bade2017gibbardsatterthwaite} \textit{constructively} characterize OSP-implementable and Pareto efficient assignment rules as the ones that can be implemented via a mechanism they call \textit{sequential barter with lurkers}. \citet{pycia2019theory} characterize the full class of OSP mechanisms in environments without transfers as \textit{millipede games with greedy strategies}. They also characterize the full class of SOSP mechanisms as \textit{sequential price mechanisms with greedy strategies}. \citet{thomas2020classification} provides a necessary and sufficient condition for an APDA rule to be OSP-implementable when there is an equal number of agents and objects. \citet{ashlagi2018stable} consider two-sided matching with one strategic side and show that for general preferences, no mechanism that implements a stable matching is obviously strategy-proof for any side of the market.

	\subsection{Organization of the paper}

	The organization of this paper is as follows. In Section \ref{section prelim}, we introduce basic notions and notations that we use throughout the paper, define assignment rules, and introduce the notions of OSP-implementation and SOSP-implementation. Section \ref{section FPTTC} introduces the notion of FPTTC rules. In Section \ref{section OSP FPTTC}, we introduce dual ownership property of an FPTTC rule and present prior results on the OSP-implementability of FPTTC rules by means of this property. We also introduce acyclicity property and strong acyclicity property of an FPTTC rule, and discuss their relationship with dual ownership property. In Section \ref{section SOSP FPTTC}, we discuss the SOSP-implementability of FPTTC rules. In Section \ref{section SSP}, we introduce the notion of simple strategy-proofness, and discuss its relation with SOSP-implementation. We further characterize all simply strategy-proof FPTTC rules. All omitted proofs are collected in the Appendix.

	\section{Preliminaries}\label{section prelim}

	\subsection{Basic notions and notations}\label{subsection notions}

	Let $N = \{1, \ldots, n\}$ be a finite set of agents, and $A$ be a non-empty and finite set of objects. Let $a_0$ denote the \textbf{\textit{outside option}}. An \textit{\textbf{allocation}} is a function $\mu: N \to A \cup \{a_0\}$ such that $|\mu^{-1}(a)| \leq 1$ for all $a \in A$. Here, $\mu(i) = a$ means agent $i$ is assigned object $a$ under $\mu$, and $\mu(i) = a_0$ means agent $i$ is not assigned any object under $\mu$. Notice that the outside option $a_0$ can be assigned to any number of agents and that not all objects in $A$ have to be assigned. We denote by $\mathcal{M}$ the set of all allocations.
	
	Let $\mathbb{L}(A \cup \{a_0\})$ denote the set of all strict linear orders over $A \cup \{a_0\}$.\footnote{A \textbf{\textit{strict linear order}} is a semiconnex, asymmetric, and transitive binary relation.} An element of $\mathbb{L}(A \cup \{a_0\})$ is called a \textbf{\textit{preference}} over $A \cup \{a_0\}$. For a preference $P$, let $R$ denote the weak part of $P$, that is, for all $a, b \in A \cup \{a_0\}$, $aRb$ if and only if \big[$aPb$ or $a = b$\big].
	For a preference $P \in \mathbb{L}(A \cup \{a_0\})$ and non-empty $A' \subseteq A \cup \{a_0\}$, let $\tau(P, A')$ denote the most preferred element in $A'$ according to $P$, that is, $\tau(P,A') = a$ if and only if \big[$a \in A'$ and $a P b$ for all $b \in A' \setminus \{a\}$\big]. For ease of presentation, we denote $\tau(P, A \cup \{a_0\})$ by $\tau(P)$.
	
	For each object $a \in A$, we define the \textbf{\textit{priority}} of $a$ as a ``preference'' $\succ_a$ over $N$.\footnote{In other words, $\succ_a \in \mathbb{L}(N)$ for all $a \in A$.} Following our notational convention, for a priority $\succ \in \mathbb{L}(N)$ and non-empty $N' \subseteq N$, let $\tau(\succ, N')$ denote the most preferred agent in $N'$ according to $\succ$. For ease of presentation, we denote $\tau(\succ, N)$ by $\tau(\succ)$. For a priority $\succ \in \mathbb{L}(N)$ and an agent $i \in N$, by $U(i, \succ)$ we denote the \textit{strict upper contour set} $\{j \in N \mid j \succ i\}$ of $i$ at $\succ$. 
	Furthermore, for a priority $\succ \in \mathbb{L}(N)$, an agent $i \in N$, and non-empty $N' \subseteq N \setminus \{i\}$, we write $i \succ N'$ to mean that $i \succ j$ for all $j \in N'$.
	
	We call a collection $\succ_A = (\succ_a)_{a \in A}$ a \textit{\textbf{priority structure}}. Let $N' \subseteq N$, $A' \subseteq A$, and $\succ_A$ be a priority structure. The \textit{\textbf{reduced priority structure $\succ_{A'}^{N'}$}} is the collection $(\succ_{a}^{N'})_{a \in A'}$ such that for all $a \in A'$, (i) $\succ_{a}^{N'} \in \mathbb{L}(N')$ and (ii) for all $i,j \in N'$, $i \succ_a^{N'} j$ if and only if $i \succ_a j$. Thus, the reduced priority structure $\succ_{A'}^{N'}$ is the restriction of $\succ_{A}$ to the reduced market $(N', A')$.\footnote{Thus, $\succ_{A}^{N} = \succ_{A}$.} Furthermore, let $\mathcal{T}(\succ_{A'}^{N'}) = \{i \mid \tau(\succ_a, N') = i \mbox{ for some } a \in A'\}$ be the set of agents who have the highest priority in $N'$ for at least one object in $A'$ according to $\succ_A$.

	\subsection{Types of domains}

	We denote by $\mathcal{P}_{i} \subseteq \mathbb{L}(A \cup \{a_0\})$ the set of admissible preferences of agent $i$. A \textbf{\textit{preference profile}}, denoted by $P_N = (P_1, \ldots,P_n)$, is an element of $\mathcal{P}_N = \underset{i=1}{\overset{n}{\prod}} \hspace{1 mm}\mathcal{P}_{i}$, that represents a collection of preferences -- one for each agent. We say an object $a$ is \textbf{\textit{acceptable}} to agent $i$ if $a P_i a_0$.
	
	Our framework encompasses both models -- with and without outside options, as special cases.
	\begin{enumerate}[(i)]
		\item \textbf{Model with outside options:} This model can be captured in our framework by taking $\mathcal{P}_{i} = \mathbb{L}(A \cup \{a_0\})$ for all $i \in N$.
		
		\item \textbf{Model without outside options:} Every object is acceptable to every agent in this model, which can be captured in our framework by considering \[\mathcal{P}_{i} = \{P \in \mathbb{L}(A \cup \{a_0\}) \mid aPa_0 \mbox{ for all } a \in A\}\] for all $i \in N$. With abuse of notation, let $\mathbb{L}(A) = \{P \in \mathbb{L}(A \cup \{a_0\}) \mid aPa_0 \mbox{ for all } a \in A\}$ be the set of all preferences where the outside option $a_0$ is the least preferred element in $A \cup \{a_0\}$.
	\end{enumerate}

	\subsection{Assignment rules and their properties}\label{subsection simple dominance}

	An \textit{\textbf{assignment rule}} is a function $f: \mathcal{P}_N \to \mathcal{M}$. For an assignment rule $f: \mathcal{P}_N \to \mathcal{M}$ and a preference profile $P_N \in \mathcal{P}_N$, let $f_i(P_N)$ denote the assignment of agent $i$ by $f$ at $P_N$.
	
	\citet{pycia2019theory} introduce the notion of \textit{simple dominance}. In this paper, we discuss two types of simple dominance in the context of assignment rules, namely \textit{obviously strategy-proofness (OSP)} and \textit{strongly obviously strategy-proofness (SOSP)}. We use the following notions and notations to present these.
	
	We denote a rooted (directed) tree by $T$. For a rooted tree $T$, we denote its set of nodes by $V(T)$, set of edges by $E(T)$, root by $r(T)$, and set of leaves (terminal nodes) by $L(T)$. For a node $v \in V(T)$, let $E^{out}(v)$ denote the set of outgoing edges from $v$. For an edge $e \in E(T)$, let $s(e)$ denote its source node. A \textit{path} in a tree is a sequence of nodes such that every two consecutive nodes form an edge.

	\begin{definition}
		An \textit{\textbf{extensive-form mechanism}}, or simply a \textit{\textbf{mechanism}} on $\mathcal{P}_N$, is defined as a tuple $G= \langle T, \eta^{LA}, \eta^{NA}, \eta^{EP} \rangle $, where 
		\begin{enumerate}[(i)]
			\item $T$ is a rooted tree,
			
			\item $\eta^{LA}: L(T) \to \mathcal{M}$ is a leaves-to-allocations function,
			
			\item $\eta^{NA}: V(T) \setminus L(T) \to N$ is a nodes-to-agents function, and
			
			\item $\eta^{EP}: E(T) \to 2^{\mathbb{L}(A \cup \{a_0\})} \setminus \{\emptyset\}$ is an edges-to-preferences function such that
			\begin{enumerate}[(a)]
				\item for all distinct $e, e' \in E(T)$ with $s(e) = s(e')$, we have $\eta^{EP}(e) \cap \eta^{EP}(e') = \emptyset$, and
				
				\item for any $v \in V(T) \setminus L(T)$,
				\begin{enumerate}[(1)]
					\item if there exists a path $(v^1, \ldots, v^t)$ from $r(T)$ to $v$ and some $1 \leq t' < t$ such that $\eta^{NA}(v^{t'}) = \eta^{NA}(v)$ and $\eta^{NA}(v^s) \neq \eta^{NA}(v)$ for all $s = t'+1, \ldots, t-1$, then $\underset{e \in E^{out}(v)}{\cup} \eta^{EP}(e) = \eta^{EP}(v^{t'}, v^{t'+1})$, and
					
					\item if there is no such path, then $\underset{e \in E^{out}(v)}{\cup} \eta^{EP}(e) = \mathcal{P}_{\eta^{NA}(v)}$.
				\end{enumerate}
			\end{enumerate}
		\end{enumerate}
	\end{definition}

	For a given mechanism $G$ on $\mathcal{P}_N$, every preference profile $P_N \in \mathcal{P}_N$ identifies a unique path from the root to some leaf in $T$ in the following manner: from each node $v$, follow the outgoing edge $e$ from $v$ such that $\eta^{EP}(e)$ contains the preference $P_{\eta^{NA}(v)}$. If a node $v$ lies in such a path, then we say that the preference profile \textit{$P_N$ passes through the node $v$}. Furthermore, we say two preferences $P_i$ and $P_i'$ of some agent $i$ \textit{diverge at a node $v \in V(T) \setminus L(T)$} if $\eta^{NA}(v) = i$ and there are two distinct outgoing edges $e$ and $e'$ in $E^{out}(v)$ such that $P_i \in \eta^{EP}(e)$ and $P'_i \in \eta^{EP}(e')$.
	
	For a given mechanism $G$ on $\mathcal{P}_N$, the \textit{\textbf{assignment rule $f^G : \mathcal{P}_N \rightarrow \mathcal{M}$ implemented by $G$}} is defined as follows: for all preference profiles $P_N \in \mathcal{P}_N$, $f^G(P_N) = \eta^{LA}(l)$, where $l$ is the leaf that appears at the end of the unique path characterized by $P_N$.

	\begin{definition}\label{def updated OSP}
		A mechanism $G$ on $\mathcal{P}_N$ is \textit{\textbf{OSP}} if for all $i \in N$, all nodes $v$ such that $\eta^{NA}(v) = i$, and all $P_N, \tilde{P}_N \in \mathcal{P}_N$ passing through $v$ such that $P_i$ and $\tilde{P}_i$ diverge at $v$, we have $f_i^G(P_N) R_i f_i^G(\tilde{P}_N)$. 
		
		An assignment rule $f : \mathcal{P}_N \rightarrow \mathcal{M}$ is \textit{\textbf{OSP-implementable}} (on $\mathcal{P}_N$) if there exists an OSP mechanism $G$ on $\mathcal{P}_N$ such that $f = f^G$.
	\end{definition}

	In other words, obvious strategy-proofness of an assignment rule implies that, whenever an agent $i$ is called to play, the worst possible final outcome from following her true preference $P_i$ is at least as good as the best possible outcome from following any deviating preference $P'_i$, where the best and worst cases are determined by considering all possible preferences of $i$'s opponents in the future, keeping her own preference fixed.

	\begin{definition}
		A mechanism $G$ on $\mathcal{P}_N$ is \textit{\textbf{SOSP}} if for all $i \in N$, all nodes $v$ such that $\eta^{NA}(v) = i$, and all $P_N, P'_N, \tilde{P}_N \in \mathcal{P}_N$ passing through $v$ such that (i) $P_i$ and $P'_i$ do not diverge at $v$ and (ii) $P_i$ and $\tilde{P}_i$ diverge at $v$, we have $f_i^G(P'_N) R_i f_i^G(\tilde{P}_N)$.\footnote{In other words, there exists an edge $e \in E^{out}(v)$ such that $P_i, P'_i \in \eta^{EP}(e)$ and $\tilde{P}_i \notin \eta^{EP}(e)$.} 
		
		An assignment rule $f : \mathcal{P}_N \rightarrow \mathcal{M}$ is \textit{\textbf{SOSP-implementable}} (on $\mathcal{P}_N$) if there exists an SOSP mechanism $G$ on $\mathcal{P}_N$ such that $f = f^G$.
	\end{definition}

	In other words, strongly obvious strategy-proofness of an assignment rule implies that, whenever an agent $i$ is called to play, the worst possible final outcome from following her true preference $P_i$ is at least as good as the best possible outcome from following any deviating preference $P'_i$, where the best and worst cases are determined by considering all possible preferences of \textit{all} agents in the future.

	\begin{remark}\label{remark SOSP implies OSP}
		For an arbitrary domain of preference profiles $\mathcal{P}_N$, every SOSP-implementable assignment rule is OSP-implementable.
	\end{remark}

	In Example \ref{example OSP but not SOSP}, we explain the difference between OSP-implementability and SOSP-implementability.

	\begin{example}\label{example OSP but not SOSP}
		Consider an allocation problem with three agents $N = \{1,2,3\}$ and three objects $A = \{a_1,a_2,a_3\}$. In Figure \ref{Tree Representation for OSP}, we provide a mechanism $G$ on $\mathbb{L}^3(A)$.\footnote{Recall that $\mathbb{L}(A) = \{P \in \mathbb{L}(A \cup \{a_0\}) \mid aPa_0 \mbox{ for all } a \in A\}$.} We use the following notations in Figure \ref{Tree Representation for OSP}: by $(v_1: 1)$ we mean that the node $v_1$ is assigned to agent $1$, by $\cdot a_1 \cdot a_2 \cdot$ we denote the set of preferences where $a_1$ is preferred to $a_2$, by $a_1 a_2 a_3 a_0$ we denote the preference that ranks $a_1$ first, $a_2$ second, $a_3$ third, and $a_0$ fourth, and we denote the allocation $[(1,a_1), (2,a_2), (3,a_3)]$ by $(a_1, a_2, a_3)$.

		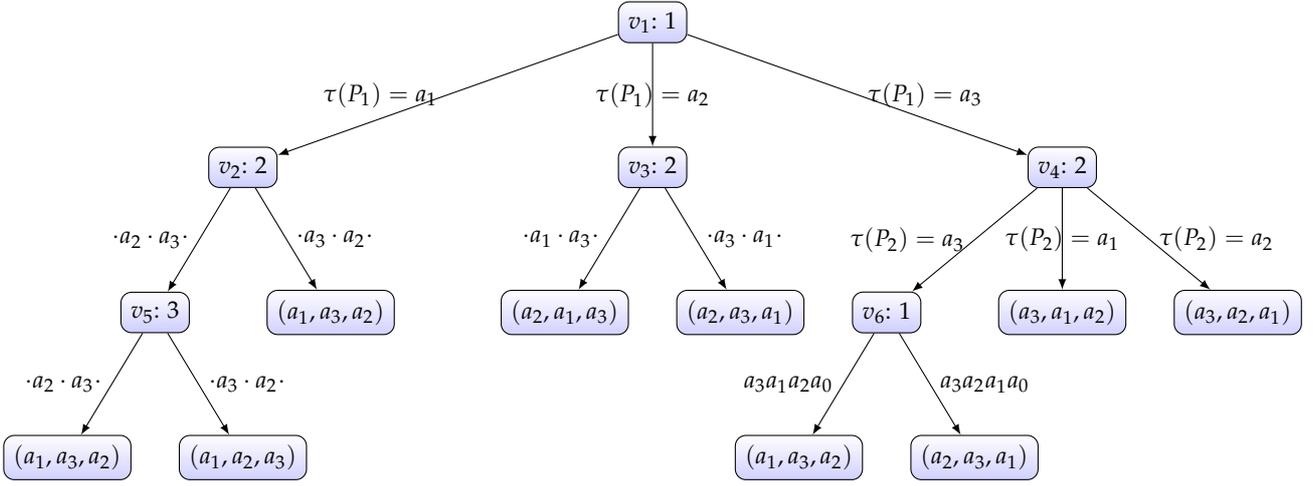
\begin{figure}[H]
			\centering
			\begin{tikzpicture}
				[
				grow                    = down,
				sibling distance        = 14em,
				level distance          = 5em,
				edge from parent/.style = {draw, -latex},
				every node/.style       = {font=\footnotesize}
				]
				\node [root] {$v_1$: 1}
				child { node [root] {$v_2$: 2}
					[
					grow                    = down,
					sibling distance        = 6em,
					level distance          = 5em,
					edge from parent/.style = {draw, -latex},
					every node/.style       = {font=\footnotesize}
					]
					child { node [root] {$v_5$: 3}
						[
						grow                    = down,
						sibling distance        = 6em,
						level distance          = 5em,
						edge from parent/.style = {draw, -latex},
						every node/.style       = {font=\footnotesize}
						]
						child { node [env] {$(a_1, a_3, a_2)$}
							edge from parent node [left] {$\cdot a_2 \cdot a_3 \cdot$} }
						child { node [env] {$(a_1, a_2, a_3)$}
							edge from parent node [right] {$\cdot a_3 \cdot a_2 \cdot$} }
						edge from parent node [left] {$\cdot a_2 \cdot a_3 \cdot$} }
					child { node [env] {$(a_1, a_3, a_2)$}
						edge from parent node [right] {$\cdot a_3 \cdot a_2 \cdot$} }
					edge from parent node [left] {$\tau(P_1) = a_1$} }
				child { node [root] {$v_3$: 2}
					[
					grow                    = down,
					sibling distance        = 6em,
					level distance          = 5em,
					edge from parent/.style = {draw, -latex},
					every node/.style       = {font=\footnotesize}
					]
					child { node [env] {$(a_2, a_1, a_3)$}
						edge from parent node [left] {$\cdot a_1 \cdot a_3 \cdot$} }
					child { node [env] {$(a_2, a_3, a_1)$}
						edge from parent node [right] {$\cdot a_3 \cdot a_1 \cdot$} }
					edge from parent node {$\tau(P_1) = a_2$} }
				child { node [root] {$v_4$: 2}
					[
					grow                    = down,
					sibling distance        = 6em,
					level distance          = 5em,
					edge from parent/.style = {draw, -latex},
					every node/.style       = {font=\footnotesize}
					]
					child { node [root] {$v_6$: 1}
						[
						grow                    = down,
						sibling distance        = 6em,
						level distance          = 5em,
						edge from parent/.style = {draw, -latex},
						every node/.style       = {font=\footnotesize}
						]
						child { node [env] {$(a_1, a_3, a_2)$}
							edge from parent node [left] {$a_3 a_1 a_2 a_0$} }
						child { node [env] {$(a_2, a_3, a_1)$}
							edge from parent node [right] {$a_3 a_2 a_1 a_0$} }
						edge from parent node [left] {$\tau(P_2) = a_3$} }
					child { node [env] {$(a_3, a_1, a_2)$}
						edge from parent node {$\tau(P_2) = a_1$} }
					child { node [env] {$(a_3, a_2, a_1)$}
						edge from parent node [right] {$\tau(P_2) = a_2$} }
					edge from parent node [right] {$\tau(P_1) = a_3$} };
			\end{tikzpicture}
			\caption{Mechanism $G$ for Example \ref{example OSP but not SOSP}}
			\label{Tree Representation for OSP}
		\end{figure}

		{
			
			\def\OldComma{,}{ }
			\catcode`\,=13
			\def,{%
				\ifmmode%
				\OldComma\discretionary{}{}{}%
				\else%
				\OldComma%
				\fi%
			}%
			
			First, we argue that $G$ is OSP. To see this, consider, for instance, the preference profiles $P_N^1 = (a_3 a_1 a_2 a_0, a_3 a_1 a_2 a_0, a_3 a_1 a_2 a_0)$ and $P_N^2 = (a_3 a_2 a_1 a_0, a_1 a_2 a_3 a_0, a_3 a_1 a_2 a_0)$.\footnote{We denote by $(a_1 a_2 a_3 a_0, a_2 a_1 a_3 a_0, a_3 a_1 a_2 a_0)$ a preference profile where agents $1, 2$, and $3$ have preferences $a_1 a_2 a_3 a_0$, $a_2 a_1 a_3 a_0$, and $a_3 a_1 a_2 a_0$, respectively.} Note that both of them pass through the node $v_4$ at which $P_2^1$ and $P_2^2$ diverge. Further note that $f^G_2(P_N^1) = a_3$ and $f^G_2(P_N^2) = a_1$, which means $G$ satisfies OSP property for this instance. Similarly, one can check that $G$ satisfies OSP property for other instances.
			
		}
		
		Now we argue that $G$ is not SOSP. Consider the preference profile $P_N = (a_3 a_1 a_2 a_0, a_3 a_1 a_2 a_0, a_3 a_1 a_2 a_0)$. Note that at node $v_1$, the worst possible final assignment for agent $1$ (according to $P_1$) from following the edge $(v_1, v_4)$ is $a_2$, whereas the best possible final assignment from not following the edge $(v_1, v_4)$ is $a_1$. Since $a_1 P_1 a_2$, this implies $G$ does not satisfy SOSP property.
		\hfill
		$\Diamond$
	\end{example}

	\section{Fixed priority top trading cycles rules}\label{section FPTTC}

	\textit{Fixed priority top trading cycles (FPTTC) rules} are well-known in the literature. We explain how such a rule works by means of an example.
	
	We begin by explaining the notion of \textit{TTC procedure} with respect to a given endowments of the objects over the agents. Suppose that each object is owned by exactly one agent. Note that an agent may own more than one object. A directed graph is constructed in the following manner. The set of nodes is the same as the set of agents. There is a directed edge from agent $i$ to agent $j$ if agent $j$ owns agent $i$'s most preferred acceptable object, and a directed edge from agent $i$ to herself if every object is unacceptable to her. Notice that each agent has exactly one outgoing edge (though possibly many incoming edges), and there is a directed edge from an agent to herself if and only if she owns her most preferred acceptable object or every object is unacceptable to her. Note that such a graph will always have at least a cycle. Such a cycle is called a \textit{top trading cycle (TTC)}. After forming a TTC, each agent in the TTC is assigned her most preferred element among all the objects and the outside option $a_0$.

	\begin{example}\label{example FPTTC}
		An FPTTC rule is based on an exogenously given priority structure. Consider an allocation problem with four agents $N = \{1,2,3,4\}$ and four objects $A = \{a_1,a_2,a_3,a_4\}$, and consider the priority structure given in Table \ref{example priority structure for example}. Here, for instance, $\succ_{a_1}$ ranks agent $1$ first, agent $4$ second, agent $3$ third, and agent $2$ fourth.
		\begin{table}[H]
			\centering
			\begin{tabular}{@{}cccc@{}}
				\hline
				$\succ_{a_1}$ & $\succ_{a_2}$ & $\succ_{a_3}$ & $\succ_{a_4}$ \\ \hline
				\hline
				$1$ & $2$ & $3$ & $4$ \\
				$4$ & $1$ & $2$ & $3$ \\
				$3$ & $3$ & $4$ & $2$ \\
				$2$ & $4$ & $1$ & $1$ \\
				\hline
			\end{tabular}
			\caption{Priority structure for Example \ref{example FPTTC}}
			\label{example priority structure for example}
		\end{table}
		
		Consider the FPTTC rule based on the priority structure given in Table \ref{example priority structure for example} and consider the preference profile $P_N$ such that $a_0 P_1 a_2 P_1 a_1 P_1 a_3 P_1 a_4$, $a_1 P_2 a_0 P_2 a_2 P_2 a_3 P_2 a_4$, $a_1 P_3 a_2 P_3 a_3 P_3 a_0 P_3 a_4$, and $a_3 P_4 a_2 P_4 a_1 P_4 a_4 P_4 a_0$. The outcome is computed through a number of steps. At each step, endowments of the agents are determined by means of the priority structure, and TTC procedure is performed with respect to the endowments.\medskip
		
		\textbf{\textit{Step 1.}}	
		At Step 1, the ``owner'' of an object $a$ is the most preferred agent according to $\succ_a$. Thus, object $a_1$ is owned by agent $1$, object $a_2$ is owned by agent $2$, object $a_3$ is owned by agent $3$, and object $a_4$ is owned by agent $4$. TTC procedure is performed with respect to these endowments to decide the outcome of Step 1, and assigned agents and assigned objects leave the market. It can be verified that for the given preference profile $P_N$, only agent $1$ is assigned at Step 1 and her assignment is the outside option $a_0$. So, only agent $1$ leaves the market.\medskip
		
		\textbf{\textit{Step 2.}}	
		The reduced market at Step 2 consists of agents $2,3,4$, and objects $a_1, a_2, a_3, a_4$. As at Step 1, the endowments of the remaining agents are decided first and then TTC procedure is performed with respect to the endowments. The owner of a remaining object $a$ is the most preferred remaining agent according to $\succ_a$. Thus, objects $a_1$ and $a_4$ are owned by agent $4$, object $a_2$ is owned by agent $2$, and object $a_3$ is owned by agent $3$. Once the endowments are decided for Step 2, TTC procedure is performed with respect to the endowments to decide the outcome of this step, and assigned agents and assigned objects leave the market. It can be verified that for the current example, agents $3$ and $4$ get objects $a_1$ and $a_3$, respectively at this step. So, agents $3$ and $4$ leave the market with objects $a_1$ and $a_3$.\medskip
		
		\textbf{\textit{Step 3.}}	
		Step 3 is followed on the reduced market in a similar way as Step 2. The reduced market at Step 3 consists of agent $2$, and objects $a_2$ and $a_4$. Both of the objects $a_2$ and $a_4$ are owned by agent $2$ at Step 3. It can be verified that agent $2$ is assigned the outside option $a_0$ at Step 3, and leaves the market.\medskip
		
		Every agent is assigned by the end of Step 3 and hence the algorithm stops at this step. Thus, agents $1$ and $2$ are assigned the outside option $a_0$, and agents $3$ and $4$ are assigned objects $a_1$ and $a_3$, respectively, at the outcome of the FPTTC rule at $P_N$.
		\hfill
		$\Diamond$	
	\end{example}

	In what follows, we present a formal description of FPTTC rules. For a given priority structure $\succ_A$, the \textit{\textbf{FPTTC rule $T^{\succ_A}$ associated with $\succ_A$}} is defined by an iterative procedure as follows. Consider an arbitrary preference profile $P_N \in \mathbb{L}^n(A \cup \{a_0\})$.

	\begin{itemize}[leftmargin = 1.35cm]
		\item[\textbf{\textit{Step s.}}] Let $N_s(P_N) \subseteq N$ be the set of agents that remain after Step $s-1$ and $A_s(P_N) \subseteq A$ be the set of objects that remain after Step $s-1$.\footnote{Notice that for all $P_N \in \mathbb{L}^n(A \cup \{a_0\})$, $N_1(P_N) = N$ and $A_1(P_N) = A$.} Each object $a \in A_s(P_N)$ is owned by its most preferred agent in $N_s(P_N)$ according to $\succ_a$. TTC procedure is performed on the reduced market $(N_s(P_N), A_s(P_N))$ with respect to these endowments. Remove all assigned agents and assigned objects at this step.
	\end{itemize}

	This procedure is repeated iteratively until either all agents are assigned or all objects are assigned. The final outcome is obtained by combining all the assignments at all steps.

	\section{OSP-implementability of FPTTC rules}\label{section OSP FPTTC}

	\citet{mandal2022obviously} study the OSP-implementability of FPTTC rules on the restricted domain $\mathbb{L}^n(A)$. For this purpose, they introduce the notion of \textit{dual ownership} and show that it is both necessary and sufficient condition for an FPTTC rule to be OSP-implementable on $\mathbb{L}^n(A)$. Later, \citet{mandal2022outside} extend this result to the unrestricted domain $\mathbb{L}^n(A \cup \{a_0\})$. In what follows, we present the notion of dual ownership, as well as the prior results on the OSP-implementability of FPTTC rules using this notion. 
	
	Before defining dual ownership formally, let us first recall a couple of notations used in the context of FPTTC rules: for a preference profile $P_N \in \mathbb{L}^n(A \cup \{a_0\})$ and an FPTTC rule, $N_s(P_N)$ is the set of remaining agents after Step $s-1$ and $A_s(P_N)$ is the set of remaining objects after Step $s-1$.

	\begin{definition}\label{def dual ownership}
		On a domain $\mathcal{P}_N$, an FPTTC rule $T^{\succ_A}$ satisfies \textbf{\textit{dual ownership}} if for all $P_N \in \mathcal{P}_N$, we have $|\mathcal{T}(\succ_{A_s(P_N)}^{N_s(P_N)})| \leq 2$ for all $s$.
	\end{definition}

	In other words, an FPTTC rule satisfies dual ownership on an arbitrary domain of preference profiles $\mathcal{P}_N$ if, for any preference profile in $\mathcal{P}_N$ and any step of the FPTTC rule at that preference profile, there are at most two agents who own all the objects that remain in the reduced market at that step.

	\begin{note}\label{remark dual ownership depends on domains}
		We define dual ownership as a property of FPTTC rules because it depends on the choice of domain. More precisely, an FPTTC rule satisfying dual ownership on a smaller domain, may not satisfy dual ownership on a larger domain (see Example \ref{example FPTTC dual on restricted not on unrestricted} for details).
	\end{note}

	\begin{example}\label{example FPTTC dual on restricted not on unrestricted}
		Consider an allocation problem with four agents $N = \{i_1,i_2,i_3,i_4\}$ and four objects $A = \{a_1,a_2,a_3,a_4\}$. Let $\succ_A$ be as follows:
		\begin{table}[H]
			\centering
			\begin{tabular}{@{}cccc@{}}
				\hline
				$\succ_{a_1}$ & $\succ_{a_2}$ & $\succ_{a_3}$ & $\succ_{a_4}$ \\ \hline
				\hline
				$i_1$ & $i_1$ & $i_4$ & $i_4$ \\
				$i_2$ & $i_2$ & $i_2$ & $i_3$ \\
				$i_3$ & $i_3$ & $i_3$ & $i_2$ \\
				$i_4$ & $i_4$ & $i_1$ & $i_1$ \\
				\hline
			\end{tabular}
			\caption{Priority structure for Example \ref{example FPTTC dual on restricted not on unrestricted}}
			\label{table for priority structure dual ownership but not dual dictatorship}
		\end{table}
		
		Consider the FPTTC rule $T^{\succ_A}$ associated with the priority structure given in Table \ref{table for priority structure dual ownership but not dual dictatorship}. First, we argue that it satisfies dual ownership on the restricted domain $\mathbb{L}^n(A)$. Since either agent $i_1$ or agent $i_4$ appears at the top position in each priority, it follows that for any preference profile in $\mathbb{L}^n(A)$, agents $i_1$ and $i_4$ will own all the objects at Step 1 of $T^{\succ_A}$. Moreover, since there are only four agents in the original market, for any preference profile in $\mathbb{L}^n(A)$, at any step from Step 3 onward of $T^{\succ_A}$, there will remain at most two agents in the corresponding reduced market and hence dual ownership will be vacuously satisfied. In what follows, we show that dual ownership will also be satisfied at Step $2$ for any preference profile in $\mathbb{L}^n(A)$. We distinguish three cases based on the possible assignments at Step 1. 
		\begin{enumerate}[(i)]
			\item\label{item counter 1} Suppose only agent $i_1$ is assigned some object at Step 1. No matter whether agent $i_1$ is assigned object $a_1$ or object $a_2$, agents $i_2$ and $i_4$ will own all the remaining objects at Step 2.
			
			\item\label{item counter 2} Suppose only agent $i_4$ is assigned some object at Step 1. 
			\begin{enumerate}
				\item If $i_4$ is assigned object $a_3$, then agents $i_1$ and $i_3$ will own all the remaining objects at Step 2.
				
				\item If $i_4$ is assigned object $a_4$, then agents $i_1$ and $i_2$ will own all the remaining objects at Step 2.		
			\end{enumerate}
			
			\item\label{item counter 3} Suppose both $i_1$ and $i_4$ are assigned some objects at Step 1. Since there are only four agents in the original market, only two agents will remain in the reduced market at Step 2. 
		\end{enumerate} 
		Since Cases \ref{item counter 1}, \ref{item counter 2}, and \ref{item counter 3} are exhaustive, it follows that $T^{\succ_A}$ satisfies dual ownership on $\mathbb{L}^n(A)$. 
		
		We now proceed to show that $T^{\succ_A}$ does not satisfy dual ownership on the unrestricted domain $\mathbb{L}^n(A \cup \{a_0\})$. Consider the preference profile $P_N \in \mathbb{L}^n(A \cup \{a_0\})$ such that $\tau(P_{i_1}) = \tau(P_{i_2}) = \tau(P_{i_3}) = a_3$ and $\tau(P_{i_4}) = a_0$. It follows from the construction of $P_N$ and the definition of $T^{\succ_A}$ that at Step 2 of $T^{\succ_A}$ at $P_N$, objects $a_1$ and $a_2$ are owned by agent $i_1$, object $a_3$ is owned by agent $i_2$, and object $a_4$ is owned by agent $i_3$, and hence $T^{\succ_A}$ violates dual ownership at $P_N$.
		\hfill
		$\Diamond$
	\end{example}

	\begin{theorem}[\citealp{mandal2022obviously, mandal2022outside}]\label{theorem prior results}
		Suppose $\mathcal{P}_N \in \Big\{ \mathbb{L}^n(A), \mathbb{L}^n(A \cup \{a_0\}) \Big\}$. On the domain $\mathcal{P}_N$, dual ownership is a necessary and sufficient condition for an FPTTC rule to be OSP-implementable.
	\end{theorem}

	\begin{remark}
		It follows from Theorem \ref{theorem prior results} and Example \ref{example FPTTC dual on restricted not on unrestricted} that the set of FPTTC rules satisfying OSP-implementability on the restricted domain $\mathbb{L}^n(A)$ is a \textit{strict} superset of those satisfying OSP-implementability on the unrestricted domain $\mathbb{L}^n(A \cup \{a_0\})$.
	\end{remark}

	Although dual ownership is an intuitive property, it is somewhat time-consuming to check whether a given FPTTC rule satisfies this property. This is because, technically, one needs to check at every preference profile whether at most two agents are owning all the (remaining) objects at every step of the FPTTC rule. In view of this observation, we introduce equivalent properties to dual ownership, which involve the priority structure only (and not the preference profiles), and thus, is more convenient to be checked.

	\subsection{Results on the restricted domain $\mathbb{L}^n(A)$}

	In this subsection, we introduce the notion of \textit{acyclicity} for FPTTC rules, and show that it is equivalent to dual ownership property on the restricted domain $\mathbb{L}^n(A)$.
	
	An FPTTC rule is acyclic if the associated priority structure does not contain any \textit{priority cycle}. We begin with a verbal description of a priority cycle. A tuple $[(i_1,i_2,i_3),(a_1,a_2,a_3)]$ where $i_1,i_2,i_3 \in N$ and $a_1,a_2,a_3 \in A$ are all distinct, constitutes a priority cycle in one of two ways. In the first way, $i_h$ is the most preferred agent of $\succ_{a_h}$ for all $h = 1, 2, 3$. To explain the second way, let us present a specific instance where agents $i_1,i_2,i_3$ and objects $a_1,a_2,a_3$ form a priority cycle. Suppose there exist distinct agents $i_4, i_5 \in N \setminus \{i_1, i_2, i_3\}$ and distinct objects $a_4, a_5 \in A \setminus \{a_1,a_2,a_3\}$. For $h = 1, \ldots, 5$, let $\succ_{a_h}$ be as given below (the dots indicate that all preferences for the corresponding parts are irrelevant and can be chosen arbitrarily).
	\begin{table}[H]
		\centering
		\begin{tabular}{@{}ccccc@{}}
			\hline
			$\succ_{a_1}$ & $\succ_{a_2}$ & $\succ_{a_3}$ & $\succ_{a_4}$ & $\succ_{a_5}$ \\ \hline
			\hline
			$i_4$ & $i_4$ & $i_4$ & $i_5$ & $i_5$ \\
			$i_1$ & $i_2$ & $i_3$ & $i_4$ & $\vdots$ \\
			$\vdots$ & $\vdots$ & $\vdots$ & $\vdots$ & $ $ \\
			\hline
		\end{tabular}
		\caption{Priority structure with a priority cycle}
		\label{table for priority cycle}
	\end{table}
	
	The priority structure in Table \ref{table for priority cycle} has the property that for all $h = 1, \ldots, 5$, the strict upper contour set of agent $i_h$ at $\succ_{a_h}$ is a subset of $\{i_4, i_5\}$. For instance, the strict upper contour set of agent $i_1$ is the singleton set $\{i_4\}$. In this case, the tuple $[(i_1,i_2,i_3),(a_1,a_2,a_3)]$ is called a priority cycle. In general, a tuple $[(i_1,i_2,i_3),(a_1,a_2,a_3)]$ is a priority cycle if one can get hold of agents $i_4, \ldots, i_t$ and objects $a_4, \ldots, a_t$ such that for all $h = 1, \ldots, t$, the strict upper contour set of agent $i_h$ at $\succ_{a_h}$ is a subset of $\{i_4, \ldots, i_t\}$. In what follows, we present a formal definition.

	\begin{definition}\label{def acyclicity}
		A tuple $[(i_1,i_2,i_3),(a_1,a_2,a_3)]$, where $i_1,i_2,i_3 \in N$ and $a_1,a_2,a_3 \in A$ are all distinct, is called a \textit{\textbf{priority cycle}} in a priority structure $\succ_A$ if either $\tau(\succ_{a_h}) = i_h$ for all $h = 1,2,3$, or there exist distinct agents $i_4, \ldots, i_t \in N \setminus \{i_1,i_2,i_3\}$ and distinct objects $a_4, \ldots, a_t \in A \setminus \{a_1,a_2,a_3\}$ such that for all $h = 1, \ldots, t$, we have $U(i_h, \succ_{a_h}) \subseteq \{i_4,\ldots,i_t\}$.\footnote{Notice that there is a ``cycle'' in a priority cycle $[(i_1,i_2,i_3),(a_1,a_2,a_3)]$ in the sense that $i_1 \succ_{a_1} \{i_2, i_3\}$, $i_2 \succ_{a_2} \{i_1, i_3\}$, and $i_3 \succ_{a_3} \{i_1, i_2\}$.}
		
		We call a priority structure \textit{\textbf{acyclic}} if it contains no priority cycles, and call an FPTTC rule \textit{\textbf{acyclic}} if it is associated with an acyclic priority structure.
	\end{definition}

	\begin{theorem}\label{theorem FPTTC is dual iff acyclic}
		On the domain $\mathbb{L}^n(A)$, an FPTTC rule $T^{\succ_A}$ satisfies dual ownership if and only if it is acyclic.
	\end{theorem}

	The proof of this theorem is relegated to Appendix \ref{appendix proof of theo FPTTC is dual iff acyclic}.
	
	Since dual ownership is both necessary and sufficient condition for an FPTTC rule to be OSP-implementable on the restricted domain $\mathbb{L}^n(A)$ (see Theorem \ref{theorem prior results}), we obtain the following corollary from Theorem \ref{theorem FPTTC is dual iff acyclic}.\footnote{As we have mentioned earlier, Theorem 1 in \citet{troyan2019obviously} is not correct. Corollary \ref{coro FPTTC OSP acyclic} is a correct version of Theorem 1 in \citet{troyan2019obviously} (in fact, our result is a general result for arbitrary values of the number of agents and the number of objects).}

	\begin{corollary}\label{coro FPTTC OSP acyclic}
		On the domain $\mathbb{L}^n(A)$, an FPTTC rule $T^{\succ_A}$ is OSP-implementable if and only if it is acyclic.
	\end{corollary}

	\subsection{Results on the unrestricted domain $\mathbb{L}^n(A \cup \{a_0\})$}

	\citet{troyan2019obviously} works on the OSP-implementable FPTTC rules on the restricted domain $\mathbb{L}^n(A)$ when there is an equal number of agents and objects. For this purpose, he introduces the notion of \textit{strong acyclicity}, which we present first.

	\begin{definition}\label{def strong acyclicity}
		A tuple $[(i_1,i_2,i_3),(a_1,a_2,a_3)]$, where $i_1,i_2,i_3 \in N$ and $a_1,a_2,a_3 \in A$ are all distinct, is called a \textit{\textbf{weak cycle}} in a priority structure $\succ_A$ if $i_1 \succ_{a_1} \{i_2, i_3\}$, $i_2 \succ_{a_2} \{i_1, i_3\}$, and $i_3 \succ_{a_3} \{i_1, i_2\}$.
		
		We call a priority structure \textit{\textbf{strongly acyclic}} if it contains no weak cycles, and call an FPTTC rule \textit{\textbf{strongly acyclic}} if it is associated with a strongly acyclic priority structure.
	\end{definition}

	\begin{remark}\label{remark priority cycle properties}
		Every strongly acyclic priority structure is acyclic. To see this, note that if a priority structure $\succ_A$ contains a priority cycle $[(i_1,i_2,i_3),(a_1,a_2,a_3)]$, then $i_1 \succ_{a_1} \{i_2, i_3\}$, $i_2 \succ_{a_2} \{i_1, i_3\}$, and $i_3 \succ_{a_3} \{i_1, i_2\}$. Therefore, every priority cycle is a weak cycle, and hence strong acyclicity implies acyclicity. However, the converse is not true (see Example \ref{example counterexample} for details).
	\end{remark}

	\begin{example}\label{example counterexample}
		Consider an allocation problem with four agents $N = \{i,j,k,l\}$ and four objects $A = \{a,b,c,d\}$. Let $\succ_A$ be as follows: 
		\begin{table}[H]
			\centering
			\begin{tabular}{@{}cccc@{}}
				\hline
				$\succ_a$ & $\succ_b$ & $\succ_c$ & $\succ_d$ \\ \hline
				\hline
				$i$ & $l$ & $l$ & $i$ \\
				$j$ & $j$ & $k$ & $j$ \\
				$k$ & $k$ & $j$ & $k$ \\
				$l$ & $i$ & $i$ & $l$ \\
				\hline
			\end{tabular}
			\caption{Priority structure for Example \ref{example counterexample}}
			\label{table for priority structure counter}
		\end{table}
		
		Note that $i \succ_a \{j, k\}$, $j \succ_b \{i, k\}$, and $k \succ_c \{i, j\}$, which means $[(i,j,k), (a,b,c)]$ is a weak cycle in $\succ_A$. Furthermore, it is straightforward to verify that $\succ_A$ is acyclic.
		\hfill
		$\Diamond$
	\end{example}

	Our next theorem says that dual ownership and strong acyclicity are equivalent properties of an FPTTC rule on the unrestricted domain $\mathbb{L}^n(A \cup \{a_0\})$.

	\begin{theorem}\label{theorem dual strong acyclicity equivalent}
		On the domain $\mathbb{L}^n(A \cup \{a_0\})$, an FPTTC rule $T^{\succ_A}$ satisfies dual ownership if and only if it is strongly acyclic.
	\end{theorem}

	The proof of this theorem is relegated to Appendix \ref{appendix proof of theo dual strong acyclicity equivalent}.
	
	Since dual ownership is both necessary and sufficient condition for an FPTTC rule to be OSP-implementable on the unrestricted domain (see Theorem \ref{theorem prior results}), we obtain the following corollary from Theorem \ref{theorem dual strong acyclicity equivalent}.

	\begin{corollary}\label{corollary FPTTC OSP strong acyclic}
		On the domain $\mathbb{L}^n(A \cup \{a_0\})$, an FPTTC rule $T^{\succ_A}$ is OSP-implementable if and only if it is strongly acyclic.
	\end{corollary}

	\subsection{Relation with Ergin-acyclicity and an illustrative corollary}

	\subsubsection{Ergin-acyclicity}

	\textit{Ergin-acyclicity} is a well-known acyclicity condition in the literature. In what follows, we present this notion, and show that it is stronger than strong acyclicity.

	\begin{definition}
		A tuple $[(i_1,i_2,i_3),(a_1,a_2)]$, where $i_1,i_2,i_3 \in N$ and $a_1,a_2 \in A$ are all distinct, is called an \textit{\textbf{Ergin-cycle}} in a priority structure $\succ_A$ if $i_1 \succ_{a_1} i_2 \succ_{a_1} i_3 \succ_{a_2} i_1$.
		
		We call a priority structure \textit{\textbf{Ergin-acyclic}} if it contains no Ergin-cycles, and call an FPTTC rule \textbf{\textit{Ergin-acyclic}} if it is associated with an Ergin-acyclic priority structure.
	\end{definition}

	\begin{proposition}\label{proposition ergin acyclicity implies strong acyclicity}
		Every Ergin-acyclic priority structure is strongly acyclic.
	\end{proposition}

	\begin{proof}[\textbf{Proof of Proposition \ref{proposition ergin acyclicity implies strong acyclicity}.}]
		We prove the contrapositive. Let $\succ_A$ be a priority structure that contains a weak cycle $[(i_1,i_2,i_3),(a_1,a_2,a_3)]$. By the definition of a weak cycle, we have $i_1 \succ_{a_1} \{i_2, i_3\}$, $i_2 \succ_{a_2} i_1$, and $i_3 \succ_{a_3} i_1$. We distinguish the following two cases.
		\begin{enumerate}[(i)]
			\item\label{item ergin 1} Suppose $i_1 \succ_{a_1} i_2 \succ_{a_1} i_3$. 
			\\
			The facts $i_1 \succ_{a_1} i_2 \succ_{a_1} i_3$ and $i_3 \succ_{a_3} i_1$ together imply that $[(i_1,i_2,i_3),(a_1,a_3)]$ is an Ergin-cycle in $\succ_A$.
			
			\item\label{item ergin 2} Suppose $i_1 \succ_{a_1} i_3 \succ_{a_1} i_2$.
			\\
			The facts $i_1 \succ_{a_1} i_3 \succ_{a_1} i_2$ and $i_2 \succ_{a_2} i_1$ together imply that $[(i_1,i_3,i_2),(a_1,a_2)]$ is an Ergin-cycle in $\succ_A$.
		\end{enumerate}
		Since Cases \ref{item ergin 1} and \ref{item ergin 2} are exhaustive, it follows that $\succ_A$ is not Ergin-acyclic. This completes the proof of Proposition \ref{proposition ergin acyclicity implies strong acyclicity}.
	\end{proof}

	\begin{remark}
		For an arbitrary domain of preference profiles $\mathcal{P}_N$, every Ergin-acyclic FPTTC rule is OSP-implementable. To see this, note that every Ergin-acyclic FPTTC rule is strongly acyclic (see Proposition \ref{proposition ergin acyclicity implies strong acyclicity}) and every strongly acyclic FPTTC rule is OSP-implementable on the unrestricted domain $\mathbb{L}^n(A \cup \{a_0\})$ (see Corollary \ref{corollary FPTTC OSP strong acyclic}). Therefore, every Ergin-acyclic FPTTC rule is OSP-implementable on any domain of preference profiles.
	\end{remark}

	It is worth noting that the converse of Proposition \ref{proposition ergin acyclicity implies strong acyclicity} is not true. Below, Example \ref{example counterexample ergin} presents a priority structure, which is strongly acyclic but not Ergin-acyclic.

	\begin{example}\label{example counterexample ergin}
		Consider an allocation problem with three agents $N = \{i,j,k\}$ and three objects $A = \{a,b,c\}$. Let $\succ_A$ be as follows: 
		\begin{table}[H]
			\centering
			\begin{tabular}{@{}ccc@{}}
				\hline
				$\succ_a$ & $\succ_b$ & $\succ_c$ \\ \hline
				\hline
				$i$ & $i$ & $j$ \\
				$j$ & $k$ & $k$ \\
				$k$ & $j$ & $i$ \\
				\hline
			\end{tabular}
			\caption{Priority structure for Example \ref{example counterexample ergin}}
			\label{table for priority structure counter ergin}
		\end{table}
		
		Note that $i \succ_b k \succ_b j$ and $j \succ_c i$, which means $[(i,k,j), (b,c)]$ is an Ergin-cycle in $\succ_A$. Furthermore, it is straightforward to verify that $\succ_A$ is strongly acyclic.
		\hfill
		$\Diamond$
	\end{example}

	\subsubsection{OSP-implementability of agent-proposing deferred acceptance rules}

	Apart from FPTTC rules, another well-known class of priority-based assignment rules is \textit{agent-proposing deferred acceptance (APDA) rules}. Below, we present a brief description of such a rule for the sake of completeness. This description follows \citet{kesten2006two}. 
	
	For a given priority structure $\succ_A$, the \textit{\textbf{APDA rule $D^{\succ_A}$ associated with $\succ_A$}} is defined by an iterative procedure as follows. Consider an arbitrary preference profile $P_N \in \mathbb{L}^n(A \cup \{a_0\})$.

	\begin{itemize}[leftmargin = 1.35cm]
		\item[\textbf{\textit{Step 1.}}] Each agent ``applies to'' her favorite element among all the objects and the outside option $a_0$. If the outside option is the favorite element of an agent, then she is assigned the outside option. If the number of agents who apply to an object, say $a$, is at least one, then the most preferred agent among them according to $\succ_a$ tentatively holds $a$. The remaining agents are ``rejected''.
		
		\item[\textbf{\textit{Step $2$.}}] Each agent, who is rejected by an object at Step $1$, applies to her next favorite element. If the outside option is the next favorite element of an agent, then she is assigned the outside option. If the total number of agents who apply to an object, say $a$, and who is tentatively holding $a$ from earlier steps, is at least one, then the most preferred agent among them according to $\succ_a$ tentatively holds $a$. The remaining agents are rejected.
	\end{itemize}

	This procedure is repeated iteratively until all agents are assigned. At this point, the tentative holding of an object becomes permanent. This completes the description of an APDA rule.\medskip
	
	\citet{ashlagi2018stable} study the OSP-implementability of APDA rules, and show that Ergin-acyclicity is a sufficient condition for an APDA rule to be OSP-implementable. We derive their result as a corollary of our results.

	\begin{corollary}[Theorem 2 in \citet{ashlagi2018stable}]\label{corollary ashlagi}
		On the domain $\mathbb{L}^n(A \cup \{a_0\})$, an APDA rule $D^{\succ_A}$ is OSP-implementable if $\succ_A$ is Ergin-acyclic.
	\end{corollary}

	\begin{proof}[\textbf{Proof of Corollary \ref{corollary ashlagi}.}]
		\citet{kesten2006two} shows that for a priority structure $\succ_A$, the associated APDA rule $D^{\succ_A}$ and the associated FPTTC rule $T^{\succ_A}$ are equivalent if and only if $\succ_A$ is Ergin-acyclic (see Theorem 1 in \citet{kesten2006two}).\footnote{\citet{kesten2006two} introduces the notion of \textit{Kesten-acyclicity} in a many-to-one matching model, and shows that it is both necessary and sufficient condition for the associated APDA rule and the associated FPTTC rule to be equivalent. It is well-known that Kesten-acyclicity and Ergin-acyclicity are equivalent in the context of one-to-one matching.} Since Ergin-acyclicity implies strong acyclicity (see Proposition \ref{proposition ergin acyclicity implies strong acyclicity}), Corollary \ref{corollary ashlagi} follows from Corollary \ref{corollary FPTTC OSP strong acyclic} and \citet{kesten2006two}'s result.
	\end{proof}

	Since strong acyclicity is weaker than Ergin-acyclicity (see Proposition \ref{proposition ergin acyclicity implies strong acyclicity}), in view of Corollary \ref{corollary ashlagi}, one may wonder about the connection between strong acyclicity and the OSP-implementability of APDA rules: \textit{Like Ergin-acyclicity, is strong acyclicity also a sufficient condition for an APDA rule to be OSP-implementable? Or is it necessary?} In what follows, we show that strong acyclicity is not a sufficient condition for an APDA rule to be OSP-implementable (Example \ref{example strong not OSP}), but is a necessary condition (Proposition \ref{proposition OSP is strong}).

	\begin{example}\label{example strong not OSP}
		Consider an allocation problem with three agents $N = \{i,j,k\}$ and three objects $A = \{a,b,c\}$. Let $\succ_A$ be as follows: 
		\begin{table}[H]
			\centering
			\begin{tabular}{@{}ccc@{}}
				\hline
				$\succ_a$ & $\succ_b$ & $\succ_c$ \\ \hline
				\hline
				$i$ & $i$ & $k$ \\
				$j$ & $j$ & $i$ \\
				$k$ & $k$ & $j$ \\
				\hline
			\end{tabular}
			\caption{Priority structure for Example \ref{example strong not OSP}}
			\label{table for priority structure strong not OSP}
		\end{table}
		
		It is straightforward to verify that $\succ_A$ is strongly acyclic. Furthermore, using similar arguments as for Proposition 1 in \citet{ashlagi2018stable}, it follows that the APDA rule $D^{\succ_A}$ is not OSP-implementable on the restricted domain $\mathbb{L}^3(A)$, and consequently is not OSP-implementable on the unrestricted domain $\mathbb{L}^3(A \cup \{a_0\})$.
		\hfill
		$\Diamond$
	\end{example}

	\begin{proposition}\label{proposition OSP is strong}
		On the domain $\mathbb{L}^n(A \cup \{a_0\})$, an APDA rule $D^{\succ_A}$ is OSP-implementable only if it is strongly acyclic.
	\end{proposition}

	The proof of this proposition is relegated to Appendix \ref{appendix proof of proposition OSP is strong}.

	\section{SOSP-implementability of FPTTC rules}\label{section SOSP FPTTC}

	Before proceeding with our next results, we first present a well-known class of assignment rules, namely \textit{serial dictatorships}. In a \textbf{\textit{serial dictatorship}}, agents are ordered, and the first agent in the ordering gets her most preferred element among all the objects and the outside option $a_0$, the second agent in the ordering gets her most preferred element among the remaining objects and the outside option $a_0$, etc.
	
	Note that serial dictatorships are special cases of FPTTC rules. For example, consider an allocation problem with three agents $N = \{1, 2, 3\}$ and three objects $A = \{a_1, a_2, a_3\}$. The priority structure associated with the FPTTC rule that corresponds to the serial dictatorship with the exogenously given ordering $(1, 2, 3)$ is as follows:
	\begin{table}[H]
		\centering
		\begin{tabular}{@{}ccc@{}}
			\hline
			$\succ_{a_1}$ & $\succ_{a_2}$ & $\succ_{a_3}$ \\ \hline
			\hline
			$1$ & $1$ & $1$ \\
			$2$ & $2$ & $2$ \\
			$3$ & $3$ & $3$ \\
			\hline
		\end{tabular}
	\end{table}
	
	Further note that serial dictatorships are SOSP-implementable. For example, consider the allocation problem with three agents $N = \{1, 2, 3\}$ and three objects $A = \{a_1, a_2, a_3\}$. On the restricted domain $\mathbb{L}^3(A)$, the SOSP mechanism in Figure \ref{tree sd are sosp} implements the serial dictatorship with the exogenously given ordering $(1, 2, 3)$.

	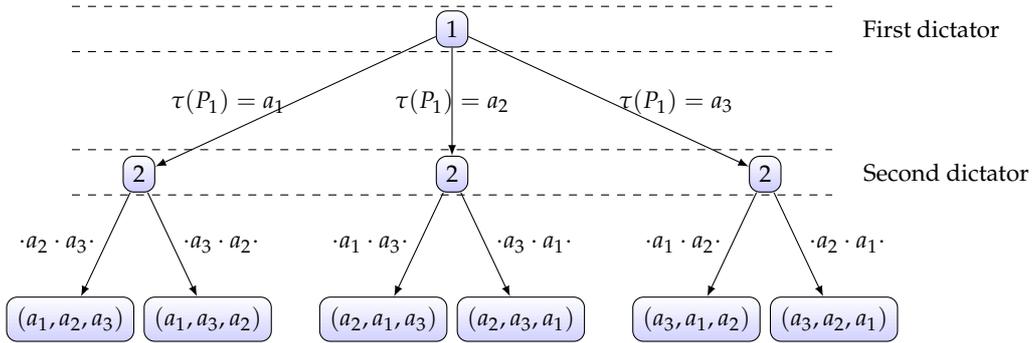
\begin{figure}[H]
		\centering
		\begin{tikzpicture}
			[
			grow                    = down,
			sibling distance        = 10.7em,
			level distance          = 5em,
			edge from parent/.style = {draw, -latex},
			every node/.style       = {font=\footnotesize}
			]
			\node [root] {1}
			child { node [root] {2}
				[
				grow                    = down,
				sibling distance        = 4.7em,
				level distance          = 5em,
				edge from parent/.style = {draw, -latex},
				every node/.style       = {font=\footnotesize}
				]
				child { node [env] {$(a_1, a_2, a_3)$}
					edge from parent node [left] {$\cdot a_2 \cdot a_3 \cdot$} }
				child { node [env] {$(a_1, a_3, a_2)$}
					edge from parent node [right] {$\cdot a_3 \cdot a_2 \cdot$} }
				edge from parent node [left] {$\tau(P_1) = a_1$} }	
			child { node [root] {2}
				[
				grow                    = down,
				sibling distance        = 4.7em,
				level distance          = 5em,
				edge from parent/.style = {draw, -latex},
				every node/.style       = {font=\footnotesize}
				]
				child { node [env] {$(a_2, a_1, a_3)$}
					edge from parent node [left] {$\cdot a_1 \cdot a_3 \cdot$} }
				child { node [env] {$(a_2, a_3, a_1)$}
					edge from parent node [right] {$\cdot a_3 \cdot a_1 \cdot$} }
				edge from parent node {$\tau(P_1) = a_2$} }
			child { node [root] {2}
				[
				grow                    = down,
				sibling distance        = 4.7em,
				level distance          = 5em,
				edge from parent/.style = {draw, -latex},
				every node/.style       = {font=\footnotesize}
				]
				child { node [env] {$(a_3, a_1, a_2)$}
					edge from parent node [left] {$\cdot a_1 \cdot a_2 \cdot$} }
				child { node [env] {$(a_3, a_2, a_1)$}
					edge from parent node [right] {$\cdot a_2 \cdot a_1 \cdot$} }
				edge from parent node [right] {$\tau(P_1) = a_3$} };
			
			\path (-5,.3) edge[dashed] node [right] {} (5,.3);
			\path (-5,-.3) edge[dashed] node [right] {} (5,-.3);
			\path (-5,-1.6) edge[dashed] node [right] {} (5,-1.6);
			\path (-5,-2.2) edge[dashed] node [right] {} (5,-2.2);
			
			\node at (6.3,0) {First dictator};
			\node at (6.5,-1.9) {Second dictator};
		\end{tikzpicture}
		\caption{SOSP mechanism that implements the serial dictatorship with ordering $(1, 2, 3)$ on $\mathbb{L}^3(A)$}
		\label{tree sd are sosp}
	\end{figure}

	Using a similar logic as for Figure \ref{tree sd are sosp}, one can construct an SOSP mechanism that implements a serial dictatorship on the unrestricted domain $\mathbb{L}^n(A \cup \{a_0\})$. We concisely sum up the above discussion as follows.

	\begin{remark}\label{remark sd are SOSP FPTTC}
		For an arbitrary domain of preference profiles $\mathcal{P}_N$, every serial dictatorship is an SOSP-implementable FPTTC rule.
	\end{remark}

	Remark \ref{remark sd are SOSP FPTTC} arises a natural question: \textit{Apart from serial dictatorships, are there any other SOSP-implementable FPTTC rules?} Theorem \ref{theorem SOSP FPTTC restricted} and Theorem \ref{theorem SOSP FPTTC unrestricted} provide answers to this question for the domains $\mathbb{L}^n(A)$ and $\mathbb{L}^n(A \cup \{a_0\})$, respectively.

	\subsection{Results on the restricted domain $\mathbb{L}^n(A)$}

	To characterize all SOSP-implementable FPTTC rules on the restricted domain $\mathbb{L}^n(A)$, we introduce a subclass of FPTTC rules, namely \textit{weak serial dictatorships}, in this subsection.

	\begin{definition}
		An FPTTC rule $T^{\succ_A}$ is a \textbf{\textit{weak serial dictatorship}} if for all $P_N \in \mathbb{L}^n(A)$ and all $s$, 
		\begin{equation*}
			|A_s(P_N)| > 2 \implies |\mathcal{T}(\succ_{A_s(P_N)}^{N_s(P_N)})| = 1.
		\end{equation*}
	\end{definition}

	In other words, an FPTTC rule is a weak serial dictatorship if, for any preference profile in $\mathbb{L}^n(A)$ and any step of the FPTTC rule at that preference profile, if there are more than two objects remaining in the reduced market at that step, then there is exactly one agent who owns all those remaining objects at that step. Notice that on the restricted domain $\mathbb{L}^n(A)$, a weak serial dictatorship has a similar procedure to a serial dictatorship as long as there are more than two objects remaining in the market.

	\begin{remark}
		Every serial dictatorship is a weak serial dictatorship.
	\end{remark}

	Our next theorem says that on the restricted domain $\mathbb{L}^n(A)$, weak serial dictatorships are the only SOSP-implementable FPTTC rules.

	\begin{theorem}\label{theorem SOSP FPTTC restricted}
		On the domain $\mathbb{L}^n(A)$, an FPTTC rule $T^{\succ_A}$ is SOSP-implementable if and only if it is a weak serial dictatorship.
	\end{theorem}

	The proof of this theorem is relegated to Appendix \ref{appendix proof of theo SOSP FPTTC restricted}.
	
	Similarly as dual ownership, to verify whether a given FPTTC rule is a weak serial dictatorship or not, one needs to check its behavior at every step at every preference profile in the restricted domain $\mathbb{L}^n(A)$. In view of this observation, we present our next result regarding the configuration of the priority structures associated with these rules. We use the following terminology to facilitate the result. For $\succ \in \mathbb{L}(N)$ and $i \in N$, we define $rank(i, \succ) = m$ if $|\{j \in N \mid j \succ i\}| = m - 1$.

	\begin{theorem}\label{theorem equivalent condition to bipolar seq}
		An FPTTC rule $T^{\succ_A}$ is a weak serial dictatorship if and only if $\succ_A$ has the following property: for all $a, b \in A$ and all $i \in N$,
		\begin{equation*}
			rank(i, \succ_a) \leq |A| - 2 \implies rank(i, \succ_a) = rank(i, \succ_b).
		\end{equation*}
	\end{theorem}

	The proof of this theorem is relegated to Appendix \ref{appendix proof of theo equivalent condition to bipolar seq}.
	
	As a corollary of Theorem \ref{theorem SOSP FPTTC restricted} and Theorem \ref{theorem equivalent condition to bipolar seq}, we obtain the configuration of the priority structures associated with the SOSP-implementable FPTTC rules on the restricted domain $\mathbb{L}^n(A)$.

	\begin{corollary}\label{corollary SOSP prio restricted}
		On the domain $\mathbb{L}^n(A)$, an FPTTC rule $T^{\succ_A}$ is SOSP-implementable if and only if $\succ_A$ has the following property: for all $a, b \in A$ and all $i \in N$,
		\begin{equation*}
			rank(i, \succ_a) \leq |A| - 2 \implies rank(i, \succ_a) = rank(i, \succ_b).
		\end{equation*}
	\end{corollary}

	We obtain the following corollary from Corollary \ref{corollary SOSP prio restricted}. It says when there are more objects than agents, SOSP-implementable FPTTC rules on the restricted domain $\mathbb{L}^n(A)$ are characterized as serial dictatorships.

	\begin{corollary}\label{corollary more objects}
		Suppose $|A| > |N|$. On the domain $\mathbb{L}^n(A)$, an FPTTC rule $T^{\succ_A}$ is SOSP-implementable if and only if it is a serial dictatorship.
	\end{corollary}

	\subsection{Results on the unrestricted domain $\mathbb{L}^n(A \cup \{a_0\})$}

	Recall that the serial dictatorships are SOSP-implementable FPTTC rules (see Remark \ref{remark sd are SOSP FPTTC}). In this subsection, we show that on the unrestricted domain $\mathbb{L}^n(A \cup \{a_0\})$, they are the \textit{only} SOSP-implementable FPTTC rules.

	\begin{theorem}\label{theorem SOSP FPTTC unrestricted}
		On the domain $\mathbb{L}^n(A \cup \{a_0\})$, an FPTTC rule $T^{\succ_A}$ is SOSP-implementable if and only if it is a serial dictatorship.
	\end{theorem}

	The proof of this theorem is relegated to Appendix \ref{appendix proof of theo SOSP FPTTC unrestricted}.

	\section{Simple strategy-proofness}\label{section SSP}

	OSP-implementation presumes that the agents understand how their own future actions will affect outcomes (the worst and the best possible final outcomes are taken only over opponents' future actions), and consequently, they might be concerned about having time-inconsistent preferences or making a mistake while performing demanding backward induction over their own future actions. \citet{pycia2019theory} introduce SOSP-implementation as a way to resolve this issue by relaxing the assumption that the agents fully comprehend how their own future actions will affect outcomes. Another way to resolve this issue will be by calling each agent to play at most once so that they need not be worried about their own future actions. In view of this observation, we introduce the notion of \textit{simple strategy-proofness} in this section.	
	
	The concept of simple strategy-proofness is based on a \textit{simple} OSP mechanism. A mechanism is simple if every agent is called to play at most once along a path. An assignment rule is \textit{simply strategy-proof} if there exists a simple OSP mechanism that implements the assignment rule. In what follows, we present formal definitions of these. Recall the definition of a mechanism $G$ given in Section \ref{subsection simple dominance}.

	\begin{definition}
		A mechanism $G = \langle T, \eta^{LA}, \eta^{NA}, \eta^{EP} \rangle $ on $\mathcal{P}_N$ is \textbf{\textit{simple}} if $\eta^{NA}(v) \neq \eta^{NA}(v')$ for all distinct $v, v' \in V(T) \setminus L(T)$ that appear in the same path.
		
		An assignment rule $f : \mathcal{P}_N \rightarrow \mathcal{M}$ is \textit{\textbf{simply strategy-proof}} (on $\mathcal{P}_N$) if there exists a simple OSP mechanism $G$ on $\mathcal{P}_N$ such that $f = f^G$.
	\end{definition}

	By definition, simple strategy-proofness is stronger than OSP-implementability. Our next result shows that simple strategy-proofness is even stronger than SOSP-implementability.

	\begin{proposition}\label{proposition SSP implies SOSP}
		For an arbitrary domain of preference profiles $\mathcal{P}_N$, every simple OSP mechanism is SOSP, and consequently, every simply strategy-proof assignment rule is SOSP-implementable.
	\end{proposition}

	The proof of this proposition is relegated to Appendix \ref{appendix proof of proposition SSP implies SOSP}.
	
	It is worth mentioning that the converse of Proposition \ref{proposition SSP implies SOSP} is not true in general. Example \ref{example sosp but not ssp} presents a domain of preference profiles and an FPTTC rule, which is SOSP-implementable but not simply strategy-proof on the given domain.

	\begin{example}\label{example sosp but not ssp}
		Consider an allocation problem with two agents $N = \{1, 2\}$ and four objects $A = \{a_1, a_2, a_3, a_4\}$. Let $\tilde{\mathcal{P}} = \{a_1 a_4 a_3 a_2 a_0, a_2 a_3 a_4 a_1 a_0, a_3 a_2 a_4 a_1 a_0, a_4 a_1 a_3 a_2 a_0\}$ be a set of preferences. Let $\succ_A$ be as follows:
		\begin{table}[H]
			\centering
			\begin{tabular}{@{}cccc@{}}
				\hline
				$\succ_{a_1}$ & $\succ_{a_2}$ & $\succ_{a_3}$ & $\succ_{a_4}$ \\ \hline
				\hline
				$1$ & $2$ & $2$ & $1$ \\
				$2$ & $1$ & $1$ & $2$ \\
				\hline
			\end{tabular}
			\caption{Priority structure for Example \ref{example sosp but not ssp}}
			\label{priority structure sosp not ssp}
		\end{table}
		
		Consider the FPTTC rule $T^{\succ_A}$ on the domain $\tilde{\mathcal{P}}^2$. It is straightforward to verify that the mechanism in Figure \ref{tree sosp not ssp} is SOSP and implements $T^{\succ_A}$ on $\tilde{\mathcal{P}}^2$.

		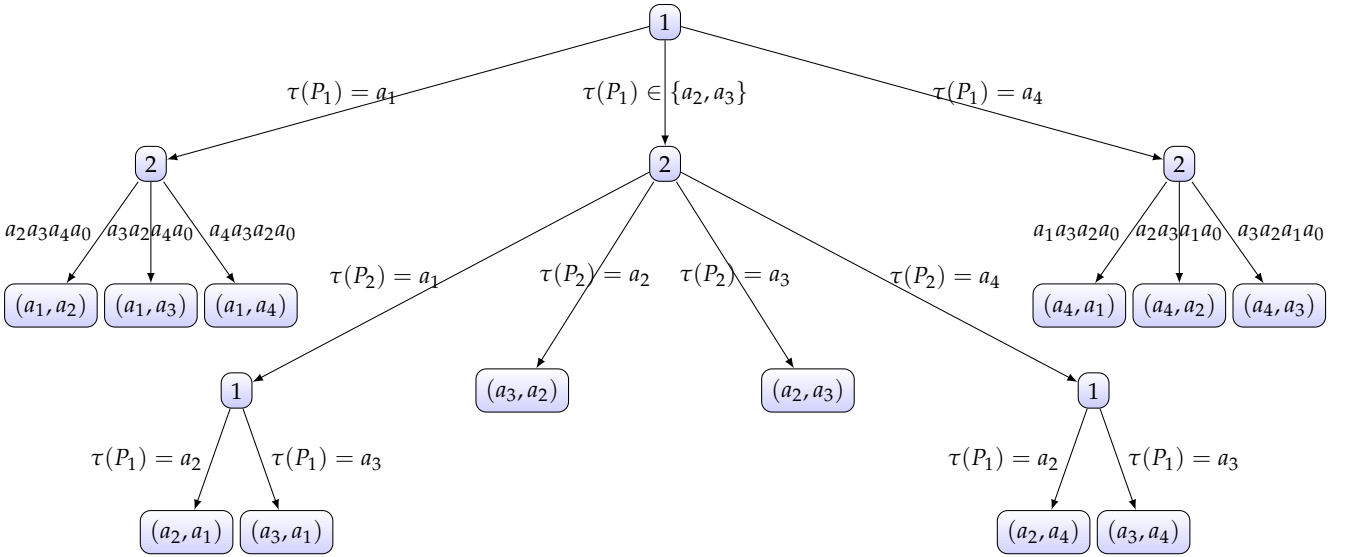
\begin{figure}[H]
			\centering
			\resizebox{\textwidth}{!}{
				\begin{tikzpicture}
					[
					grow                    = down,
					sibling distance        = 18em,
					level distance          = 5em,
					edge from parent/.style = {draw, -latex},
					every node/.style       = {font=\footnotesize}
					]
					\node [root] {1}
					child { node [root] {2}
						[
						grow                    = down,
						sibling distance        = 3.5em,
						level distance          = 5em,
						edge from parent/.style = {draw, -latex},
						every node/.style       = {font=\footnotesize}
						]
						child { node [env] {$(a_1, a_2)$}
							edge from parent node [left] {$a_2 a_3 a_4 a_0$} }
						child { node [env] {$(a_1, a_3)$}
							edge from parent node {$a_3 a_2 a_4 a_0$} }
						child { node [env] {$(a_1, a_4)$}
							edge from parent node [right] {$a_4 a_3 a_2 a_0$} }
						edge from parent node [left] {$\tau(P_1) = a_1$} }	
					child { node [root] {2}
						[
						grow                    = down,
						sibling distance        = 10em,
						level distance          = 8em,
						edge from parent/.style = {draw, -latex},
						every node/.style       = {font=\footnotesize}
						]
						child { node [root] {1}
							[
							grow                    = down,
							sibling distance        = 3.5em,
							level distance          = 5em,
							edge from parent/.style = {draw, -latex},
							every node/.style       = {font=\footnotesize}
							]
							child { node [env] {$(a_2, a_1)$}
								edge from parent node [left] {$\tau(P_1) = a_2$} }
							child { node [env] {$(a_3, a_1)$}
								edge from parent node [right] {$\tau(P_1) = a_3$} }
							edge from parent node [left] {$\tau(P_2) = a_1$} }
						child { node [env] {$(a_3, a_2)$}
							edge from parent node {$\tau(P_2) = a_2$} }
						child { node [env] {$(a_2, a_3)$}
							edge from parent node {$\tau(P_2) = a_3$} }
						child { node [root] {1}
							[
							grow                    = down,
							sibling distance        = 3.5em,
							level distance          = 5em,
							edge from parent/.style = {draw, -latex},
							every node/.style       = {font=\footnotesize}
							]
							child { node [env] {$(a_2, a_4)$}
								edge from parent node [left] {$\tau(P_1) = a_2$} }
							child { node [env] {$(a_3, a_4)$}
								edge from parent node [right] {$\tau(P_1) = a_3$} }
							edge from parent node [right] {$\tau(P_2) = a_4$} }
						edge from parent node {$\tau(P_1) \in \{a_2, a_3\}$} }
					child { node [root] {2}
						[
						grow                    = down,
						sibling distance        = 3.5em,
						level distance          = 5em,
						edge from parent/.style = {draw, -latex},
						every node/.style       = {font=\footnotesize}
						]
						child { node [env] {$(a_4, a_1)$}
							edge from parent node [left] {$a_1 a_3 a_2 a_0$} }
						child { node [env] {$(a_4, a_2)$}
							edge from parent node {$a_2 a_3 a_1 a_0$} }
						child { node [env] {$(a_4, a_3)$}
							edge from parent node [right] {$a_3 a_2 a_1 a_0$} }
						edge from parent node [right] {$\tau(P_1) = a_4$} };
			\end{tikzpicture}}
			\caption{Tree Representation for Example \ref{example sosp but not ssp}}
			\label{tree sosp not ssp}
		\end{figure}

		We argue that $T^{\succ_A}$ is not simply strategy-proof on $\tilde{\mathcal{P}}^2$. Consider the preference profiles (together with the outcomes of $T^{\succ_A}$) presented in Table \ref{preference choice not SSP}.
		\begin{table}[H]
			\centering
			\begin{tabular}{@{}c|cc|cc@{}}
				\hline
				$\mbox{Preference profiles}$ & $\mbox{Agent } 1$ & $\mbox{Agent } 2$ & $T^{\succ_A}_1$ & $T^{\succ_A}_2$ \\ \hline
				\hline
				$\tilde{P}_N^1$ & $a_2 a_3 a_4 a_1 a_0$ & $a_2 a_3 a_4 a_1 a_0$ & $a_3$ & $a_2$ \\ \hline
				$\tilde{P}_N^2$ & $a_3 a_2 a_4 a_1 a_0$ & $a_3 a_2 a_4 a_1 a_0$ & $a_2$ & $a_3$ \\ \hline
				$\tilde{P}_N^3$ & $a_3 a_2 a_4 a_1 a_0$ & $a_4 a_1 a_3 a_2 a_0$ & $a_3$ & $a_4$ \\ \hline
				$\tilde{P}_N^4$ & $a_2 a_3 a_4 a_1 a_0$ & $a_4 a_1 a_3 a_2 a_0$ & $a_2$ & $a_4$ \\ \hline
				$\tilde{P}_N^5$ & $a_1 a_4 a_3 a_2 a_0$ & $a_1 a_4 a_3 a_2 a_0$ & $a_1$ & $a_4$ \\ \hline
				$\tilde{P}_N^6$ & $a_2 a_3 a_4 a_1 a_0$ & $a_1 a_4 a_3 a_2 a_0$ & $a_2$ & $a_1$ \\ \hline
				$\tilde{P}_N^7$ & $a_4 a_1 a_3 a_2 a_0$ & $a_4 a_1 a_3 a_2 a_0$ & $a_4$ & $a_1$ \\ \hline
			\end{tabular}
			\caption{Preference profiles for Example \ref{example sosp but not ssp}}
			\label{preference choice not SSP}
		\end{table}
		
		Assume for contradiction that $T^{\succ_A}$ is simply strategy-proof on $\tilde{\mathcal{P}}^2$. So, there exists a simple OSP mechanism $\tilde{G}$ that implements $T^{\succ_A}$ on $\tilde{\mathcal{P}}_N$. Note that since $T^{\succ_A}(\tilde{P}_N^1) \neq T^{\succ_A}(\tilde{P}_N^2)$, there exists a node in the simple OSP mechanism $\tilde{G}$ that has at least two edges. Consider the first node (from the root) $v$ that has at least two edges. We distinguish the following two cases.
		\begin{enumerate}[(i)]
			\item\label{item not SSP 1} Suppose $\eta^{NA}(v) = 1$. 
			\\
			By obvious strategy-proofness of $\tilde{G}$, the facts $a_2 \tilde{P}_1^1 a_3$, $T^{\succ_A}_1(\tilde{P}_N^1) = a_3$, and $T^{\succ_A}_1(\tilde{P}_N^2) = a_2$ together imply that $\tilde{P}_1^1$ and $\tilde{P}_1^2$ do not diverge at $v$. This, together with the facts that $\tilde{P}_1^1 = \tilde{P}_1^4$, $\tilde{P}_1^2 = \tilde{P}_1^3$, $\tilde{P}_2^3 = \tilde{P}_2^4$, and $T^{\succ_A}(\tilde{P}_N^3) \neq T^{\succ_A}(\tilde{P}_N^4)$, implies that there exists a node $v'$ at which $\tilde{P}_1^1$ and $\tilde{P}_1^2$ diverge. Clearly, $v$ and $v'$ are distinct nodes appearing in the same path such that $\eta^{NA}(v) = \eta^{NA}(v') = 1$. This contradicts the fact that $\tilde{G}$ is a simple mechanism.
			
			\item\label{item not SSP 2} Suppose $\eta^{NA}(v) = 2$. 
			\\
			By obvious strategy-proofness of $\tilde{G}$, the facts $a_1 \tilde{P}_2^5 a_4$, $T^{\succ_A}_2(\tilde{P}_N^5) = a_4$, and $T^{\succ_A}_2(\tilde{P}_N^7) = a_1$ together imply that $\tilde{P}_2^5$ and $\tilde{P}_2^7$ do not diverge at $v$. This, together with the facts that $\tilde{P}_2^5 = \tilde{P}_2^6$, $\tilde{P}_2^7 = \tilde{P}_2^4$, $\tilde{P}_1^4 = \tilde{P}_1^6$, and $T^{\succ_A}(\tilde{P}_N^4) \neq T^{\succ_A}(\tilde{P}_N^6)$, implies that there exists a node $v'$ at which $\tilde{P}_2^5$ and $\tilde{P}_2^7$ diverge. Clearly, $v$ and $v'$ are distinct nodes appearing in the same path such that $\eta^{NA}(v) = \eta^{NA}(v') = 2$. This contradicts the fact that $\tilde{G}$ is a simple mechanism.
		\end{enumerate}
		Since Cases \ref{item not SSP 1} and \ref{item not SSP 2} are exhaustive, it follows that $T^{\succ_A}$ is not simply strategy-proof on $\tilde{\mathcal{P}}^2$.
		\hfill
		$\Diamond$
	\end{example}

	Example \ref{example sosp but not ssp} shows that on an arbitrary domain of preference profiles, every SOSP-implementable FPTTC rule might not be simply strategy-proof. However, our next result shows that on the restricted domain $\mathbb{L}^n(A)$ and the unrestricted domain $\mathbb{L}^n(A \cup \{a_0\})$, every SOSP-implementable FPTTC rule is simply strategy-proof.

	\begin{theorem}\label{theorem SSP SOSP FPTTC}
		Suppose $\mathcal{P}_N \in \Big\{ \mathbb{L}^n(A), \mathbb{L}^n(A \cup \{a_0\}) \Big\}$. On the domain $\mathcal{P}_N$, an FPTTC rule $T^{\succ_A}$ is simply strategy-proof if and only if it is SOSP-implementable.
	\end{theorem}

	The proof of this proposition is relegated to Appendix \ref{appendix proof of theorem SSP SOSP FPTTC}.
	
	We obtain the following corollary from Theorem \ref{theorem SOSP FPTTC restricted}, Theorem \ref{theorem SOSP FPTTC unrestricted}, and Theorem \ref{theorem SSP SOSP FPTTC}.

	\begin{corollary}\label{corollary SSP FPTTC}
		\begin{enumerate}[(a)]
			\item On the domain $\mathbb{L}^n(A)$, an FPTTC rule $T^{\succ_A}$ is simply strategy-proof if and only if it is a weak serial dictatorship.
			
			\item On the domain $\mathbb{L}^n(A \cup \{a_0\})$, an FPTTC rule $T^{\succ_A}$ is simply strategy-proof if and only if it is a serial dictatorship.
		\end{enumerate}
	\end{corollary}

	Our results show that while the class of simply strategy-proof FPTTC rules is generally a subset of the class of SOSP-implementable FPTTC rules, the ``richness'' of the domains $\mathbb{L}^n(A)$ and $\mathbb{L}^n(A \cup \{a_0\})$ makes these two classes equivalent. In view of this observation, one may wonder whether the richness drives this equivalence beyond FPTTC rules. In Appendix \ref{appendix examples}, we show that on both of these domains, not every SOSP-implementable assignment rule is simply strategy-proof.

	\appendixtocon
	\appendixtitletocon

	\renewcommand{\theequation}{\Alph{section}.\arabic{equation}}
	\renewcommand{\thetable}{\Alph{section}.\arabic{table}}
	\renewcommand{\thefigure}{\Alph{section}.\arabic{figure}}

	\begin{appendices}

		\section{Some additional notations}
		
		\setcounter{equation}{0}
		\setcounter{table}{0}
		\setcounter{figure}{0}

		It will be convenient to introduce some additional notations for the proofs. Following our notational terminology in Section \ref{subsection notions}, for a preference $P \in \mathbb{L}(A \cup \{a_0\})$ and two disjoint subsets $A'$ and $\hat{A}$ of $A \cup \{a_0\}$, we write $A' P \hat{A}$ to mean that $a P b$ for all $a \in A'$ and all $b \in \hat{A}$. Furthermore, for a preference profile $P_N$ and an FPTTC rule, let $I_s(P_N)$ be the set of assigned agents at Step $s$, $I^s(P_N)$ be the set of assigned agents up to Step $s$ (including Step $s$), $X_s(P_N)$ be the set of assigned objects at Step $s$, $X^s(P_N)$ be the set of assigned objects up to Step $s$ (including Step $s$), and $O_s(i,P_N)$ be the set of objects owned by agent $i$ at Step $s$.

		\section{Proof of Theorem \ref{theorem FPTTC is dual iff acyclic}}\label{appendix proof of theo FPTTC is dual iff acyclic}
		
		\setcounter{equation}{0}
		\setcounter{table}{0}
		\setcounter{figure}{0}

		\textbf{\textit{(If part)}} Suppose $T^{\succ_A}$ does not satisfy dual ownership on $\mathbb{L}^n(A)$. We show that $\succ_A$ contains a priority cycle. Since $T^{\succ_A}$ does not satisfy dual ownership on $\mathbb{L}^n(A)$, there exist a preference profile $\tilde{P}_N \in \mathbb{L}^n(A)$ and a step $s^*$ of $T^{\succ_A}$ at $\tilde{P}_N$ such that $|\mathcal{T}(\succ_{A_{s^*}(\tilde{P}_N)}^{N_{s^*}(\tilde{P}_N)})| > 2$. This implies that there exist three agents $i_1, i_2, i_3 \in N_{s^*}(\tilde{P}_N)$ and three objects $a_1, a_2, a_3 \in A_{s^*}(\tilde{P}_N)$ such that for all $h = 1,2,3$, agent $i_h$ owns the object $a_h$ at Step $s^*$. We proceed to show that $[(i_1, i_2, i_3), (a_1, a_2, a_3)]$ is a priority cycle in $\succ_A$. We distinguish the following two cases.\medskip
		\\
		\noindent\textbf{\textsc{Case} 1}: Suppose $s^* = 1$.
		
		Since for all $h = 1, 2, 3$, agent $i_h$ owns the object $a_h$ at Step $1$, by the definition of $T^{\succ_A}$, it follows that $\tau(\succ_{a_h}) = i_h$ for all $h = 1,2,3$. This means $[(i_1, i_2, i_3), (a_1, a_2, a_3)]$ is a priority cycle in $\succ_A$.\medskip
		\\
		\noindent\textbf{\textsc{Case} 2}: Suppose $s^* > 1$.
		
		Let $\{i_4, \ldots, i_t\} \subseteq N \setminus \{i_1,i_2,i_3\}$ and $\{a_4,\ldots,a_t\} \subseteq A \setminus \{a_1,a_2,a_3\}$ be as follows. 
		\begin{enumerate}[(i)]
			\item $\{i_4, \ldots, i_t\} = I^{s^*-1}(\tilde{P}_N)$.
			
			\item For all $h = 4, \ldots, t$, $\{a_h\} = \big( X_s(\tilde{P}_N) \cap O_s(i_h,\tilde{P}_N) \big)$ where $i_h \in I_s(\tilde{P}_N)$ for some $s < s^*$. To see that this is well-defined note that by the definition of $T^{\succ_A}$ and the fact $\tilde{P}_N \in \mathbb{L}^n(A)$, (a) for every $i_h \in I^{s^*-1}(\tilde{P}_N)$, there exists exactly one step $s$ with $s < s^*$ such that $i_h \in I_s(\tilde{P}_N)$, and (b) $O_s(i_h, \tilde{P}_N) \cap X_s(\tilde{P}_N)$ is a singleton set for all $i_h \in I_s(\tilde{P}_N)$ with $s < s^*$. 
		\end{enumerate}	
		
		It follows from the definition of $T^{\succ_A}$ and the construction of $\{i_4, \ldots, i_t\}$ and $\{a_4,\ldots,a_t\}$ that $U(i_h, \succ_{a_h}) \subseteq \{i_4,\ldots,i_t\}$ for all $h = 1, \ldots, t$.
		This implies that $[(i_1, i_2, i_3), (a_1, a_2, a_3)]$ is a priority cycle in $\succ_A$, which completes the proof of the ``if'' part of Theorem \ref{theorem FPTTC is dual iff acyclic}.\bigskip
		\\
		\textbf{\textit{(Only-if part)}} Suppose $\succ_A$ contains a priority cycle $[(i_1, i_2, i_3), (a_1, a_2, a_3)]$. We show that $T^{\succ_A}$ does not satisfy dual ownership on $\mathbb{L}^n(A)$. By the definition of a priority cycle, one of the following two statements must hold.
		\begin{enumerate}[(1)]
			\item\label{item theorem FPTTC is dual iff acyclic if 1} $\tau(\succ_{a_h}) = i_h$ for all $h = 1,2,3$.
			
			\item\label{item theorem FPTTC is dual iff acyclic if 2} There exist distinct agents $i_4, \ldots, i_t \in N \setminus \{i_1,i_2,i_3\}$ and distinct objects $a_4, \ldots, a_t \in A \setminus \{a_1,a_2,a_3\}$ such that for all $h = 1, \ldots, t$, we have $U(i_h, \succ_{a_h}) \subseteq \{i_4,\ldots,i_t\}$.
		\end{enumerate}
		We distinguish the following two cases.\medskip
		\\
		\noindent\textbf{\textsc{Case} 1}: Suppose \ref{item theorem FPTTC is dual iff acyclic if 1} holds.
		
		Since $\tau(\succ_{a_h}) = i_h$ for all $h = 1,2,3$, it must be that for any preference profile, agents $i_1, i_2$, and $i_3$ own objects $a_1, a_2$, and $a_3$, respectively, at Step $1$ of $T^{\succ_A}$ at that preference profile. Therefore $T^{\succ_A}$ does not satisfy dual ownership on $\mathbb{L}^n(A)$.\medskip
		\\
		\noindent\textbf{\textsc{Case} 2}: Suppose \ref{item theorem FPTTC is dual iff acyclic if 2} holds.
		
		Consider the preference profile $\tilde{P}_N \in \mathbb{L}^n(A)$ defined as follows. Each $i_h \in \{i_4,\ldots,i_t\}$ has a preference $\tilde{P}_{i_h}$ such that $\tau(\tilde{P}_{i_h}) = a_h$ and each $j \in N \setminus \{i_4,\ldots,i_t\}$ has a preference $\tilde{P}_j$ such that $\{a_4,\ldots,a_t\} \tilde{P}_j (A \setminus \{a_4,\ldots,a_t\})$. The next claim establishes some properties of the outcome of $T^{\succ_A}$ at $\tilde{P}_N$ at Step 1.

		\begin{claim}\label{claim for step 1}
			(a) $I_1(\tilde{P}_N) \subseteq \{i_4,\ldots,i_t\}$, and (b) $T^{\succ_A}_{i_h}(\tilde{P}_N) = a_h$ for all $h = 4, \ldots, t$ with $i_h \in I_1(\tilde{P}_N)$.
		\end{claim}

		\begin{claimproof}[\textbf{Proof of Claim \ref{claim for step 1}.}]
			By the assumptions for Case 2, it follows that $\{a_4, \ldots, a_t\} \subseteq \underset{h = 4}{\overset{t}{\cup}} \hspace{1 mm} O_1(i_h,\tilde{P}_N)$. Moreover, by the construction of $\tilde{P}_N$, we have $\tau(\tilde{P}_i) \in \{a_4, \ldots, a_t\}$ for all $i \in N$. Since $\{a_4, \ldots, a_t\} \subseteq \underset{h = 4}{\overset{t}{\cup}} \hspace{1 mm} O_1(i_h,\tilde{P}_N)$ and $\tau(\tilde{P}_i) \in \{a_4, \ldots, a_t\}$ for all $i \in N$, it follows from the definition of $T^{\succ_A}$ that $I_1(\tilde{P}_N) \subseteq \{i_4,\ldots,i_t\}$ and $T^{\succ_A}_{i}(\tilde{P}_N) = \tau(\tilde{P}_i)$ for all $i \in I_1(\tilde{P}_N)$. These two facts, along with the construction of $\tilde{P}_N$, complete the proof of Claim \ref{claim for step 1}.
		\end{claimproof}

		By Claim \ref{claim for step 1}, $I_1(\tilde{P}_N) \subseteq \{i_4,\ldots,i_t\}$ and $T^{\succ_A}_{i_h}(\tilde{P}_N) = a_h$ for all $h = 4, \ldots, t$ with $i_h \in I_1(\tilde{P}_N)$. We proceed to show that there will be a step $s^*$ such that $I^{s^*}(\tilde{P}_N) = \{i_4,\ldots,i_t\}$ and $T^{\succ_A}_{i_h}(\tilde{P}_N) = a_h$ for all $h = 4, \ldots, t$. If $I_1(\tilde{P}_N) = \{i_4,\ldots,i_t\}$, then $s^* = 1$ and we are done. Suppose $I_1(\tilde{P}_N) \subsetneq \{i_4,\ldots,i_t\}$, that is, $I_1(\tilde{P}_N)$ is a proper subset of $\{i_4,\ldots,i_t\}$. Since $I_1(\tilde{P}_N) \subsetneq \{i_4,\ldots,i_t\}$ and $T^{\succ_A}_{i_h}(\tilde{P}_N) = a_h$ for all $h = 4, \ldots, t$ with $i_h \in I_1(\tilde{P}_N)$, using similar argument as for Claim \ref{claim for step 1}, it follows from the assumptions for Case 2 and the construction of $\tilde{P}_N$ that $I_2(\tilde{P}_N) \subseteq \big( \{i_4,\ldots,i_t\} \setminus I_1(\tilde{P}_N) \big)$ and $T^{\succ_A}_{i_h}(\tilde{P}_N) = a_h$ for all $h = 4, \ldots, t$ with $i_h \in I_2(\tilde{P}_N)$. If $I_1(\tilde{P}_N) \cup I_2(\tilde{P}_N) = \{i_4,\ldots,i_t\}$, then $s^* = 2$ and we are done. Otherwise, continuing in this manner, we obtain a step $s^* > 2$ of $T^{\succ_A}$ at $\tilde{P}_N$ such that $I^{s^*}(\tilde{P}_N) = \{i_4,\ldots,i_t\}$ and $T^{\succ_A}_{i_h}(\tilde{P}_N) = a_h$ for all $h = 4, \ldots, t$.
		
		Since $I^{s^*}(\tilde{P}_N) = \{i_4,\ldots,i_t\}$ and $T^{\succ_A}_{i_h}(\tilde{P}_N) = a_h$ for all $h = 4, \ldots, t$, by the assumptions for Case 2, we have $a_h \in O_{s^*+1}(i_h, \tilde{P}_N)$ for all $h = 1,2,3$. This implies that agents $i_1$, $i_2$, and $i_3$ own the objects $a_1$, $a_2$, and $a_3$, respectively, at Step $s^* + 1$ of $T^{\succ_A}$ at $\tilde{P}_N$. Therefore $T^{\succ_A}$ does not satisfy dual ownership on $\mathbb{L}^n(A)$, which completes the proof of the ``only-if'' part of Theorem \ref{theorem FPTTC is dual iff acyclic}.
		\hfill
		\qed

		\section{Proof of Theorem \ref{theorem dual strong acyclicity equivalent}}\label{appendix proof of theo dual strong acyclicity equivalent}
		
		\setcounter{equation}{0}
		\setcounter{table}{0}
		\setcounter{figure}{0}

		Before we start proving Theorem \ref{theorem dual strong acyclicity equivalent}, to facilitate the proof we present the notion of \textit{dual dictatorship} \citep{troyan2019obviously}.

		\begin{definition}
			On a domain $\mathcal{P}_N$, an FPTTC rule $T^{\succ_A}$ satisfies \textbf{\textit{dual dictatorship}} if, for all $N' \subseteq N$ and all $A' \subseteq A$, we have $|\mathcal{T}(\succ_{A'}^{N'})| \leq 2$.
		\end{definition}

		Notice that unlike dual ownership, dual dictatorship does not depend on the choice of the domain.

		\begin{proof}[\textbf{Completion of the proof of Theorem \ref{theorem dual strong acyclicity equivalent}.}]
			\citet{mandal2022outside} show that dual ownership and dual dictatorship are equivalent properties of an FPTTC rule on the unrestricted domain (see Theorem 4.1 in \citet{mandal2022outside}).\footnote{For an arbitrary domain of preference profiles $\mathcal{P}_N$, the set of FPTTC rules satisfying dual ownership is a superset of those satisfying dual dictatorship. See \citet{mandal2022outside} for a detailed discussion about the relation between dual ownership and dual dictatorship.} 
			In a setting with \textit{equal} number of agents and objects, \citet{troyan2019obviously} shows that dual dictatorship and strong acyclicity are equivalent properties of an FPTTC rule on the restricted domain $\mathbb{L}^n(A)$ (see Theorem 2 in \citet{troyan2019obviously}). 
			His proof works verbatim on the unrestricted domain $\mathbb{L}^n(A \cup \{a_0\})$ in our setting (that is, with arbitrary values of the number of agents and the number of objects), and his result still holds. Combining all these facts, we obtain that dual ownership, dual dictatorship, and strong acyclicity are equivalent properties of an FPTTC rule on $\mathbb{L}^n(A \cup \{a_0\})$. This completes the proof of Theorem \ref{theorem dual strong acyclicity equivalent}.
		\end{proof}

		\section{Proof of Proposition \ref{proposition OSP is strong}}\label{appendix proof of proposition OSP is strong}
		
		\setcounter{equation}{0}
		\setcounter{table}{0}
		\setcounter{figure}{0}

		Let $D^{\succ_A}$ be an OSP-implementable APDA rule on the unrestricted domain $\mathbb{L}^n(A \cup \{a_0\})$. Assume for contradiction that $\succ_A$ contains a weak cycle, say $[(i_1,i_2,i_3),(a_1,a_2,a_3)]$. By the definition of a weak cycle, we have $i_1 \succ_{a_1} \{i_2, i_3\}$, $i_2 \succ_{a_2} \{i_1, i_3\}$, and $i_3 \succ_{a_3} \{i_1, i_2\}$.
		
		Fix a preference $\hat{P} \in \mathbb{L}(A \cup \{a_0\})$ such that $\tau(\hat{P}) = a_0$. We distinguish the following cases.\medskip
		\\
		\noindent\textbf{\textsc{Case} 1}: Suppose $i_1 \succ_{a_1} i_2 \succ_{a_1} i_3$ and $i_2 \succ_{a_2} i_1 \succ_{a_2} i_3$.
		
		Consider the domain $\tilde{\mathcal{P}}_N \subseteq \mathbb{L}^n(A \cup \{a_0\})$ with only four preference profiles presented (together with the outcomes of $D^{\succ_A}$) in Table \ref{preference choice OSP implies strong case 1}. Here, $l$ denotes an agent (might be empty) other than $i_1, i_2$, and $i_3$. Note that such an agent does not change her preference across the mentioned preference profiles.
		\begin{table}[H]
			\centering
			\begin{tabular}{@{}c|ccccc|ccc@{}}
				\hline
				$\mbox{Preference profiles}$ & $\mbox{Agent } i_1$ & $\mbox{Agent } i_2$ & $\mbox{Agent } i_3$ & $\ldots$ & $\mbox{Agent } l$ & $D^{\succ_A}_{i_1}$ & $D^{\succ_A}_{i_2}$ & $D^{\succ_A}_{i_3}$ \\ \hline
				\hline
				$\tilde{P}_N^1$ & $a_2 a_1 a_3 a_0\ldots$ & $a_3 a_2 a_1 a_0\ldots$ & $a_2 a_1 a_3 a_0\ldots$ & $\ldots$ & $\hat{P}$ & $a_2$ & $a_3$ & $a_1$ \\ \hline
				$\tilde{P}_N^2$ & $a_3 a_1 a_2 a_0\ldots$ & $a_1 a_3 a_2 a_0\ldots$ & $a_1 a_3 a_2 a_0\ldots$ & $\ldots$ & $\hat{P}$ & $a_1$ & $a_2$ & $a_3$ \\ \hline
				$\tilde{P}_N^3$ & $a_3 a_2 a_1 a_0\ldots$ & $a_1 a_3 a_2 a_0\ldots$ & $a_2 a_1 a_3 a_0\ldots$ & $\ldots$ & $\hat{P}$ & $a_3$ & $a_1$ & $a_2$ \\ \hline
				$\tilde{P}_N^4$ & $a_3 a_2 a_1 a_0\ldots$ & $a_3 a_2 a_1 a_0\ldots$ & $a_2 a_3 a_1 a_0\ldots$ & $\ldots$ & $\hat{P}$ & $a_1$ & $a_2$ & $a_3$ \\ \hline
			\end{tabular}
			\caption{Preference profiles of $\tilde{\mathcal{P}}_N$ for Case 1}
			\label{preference choice OSP implies strong case 1}
		\end{table}
		
		Since $D^{\succ_A}$ is OSP-implementable on $\mathbb{L}^n(A\cup \{a_0\})$, it must be OSP-implementable on the domain $\tilde{\mathcal{P}}_N$. Let $\tilde{G}$ be an OSP mechanism that implements $D^{\succ_A}$ on $\tilde{\mathcal{P}}_N$. 
		
		Note that since $D^{\succ_A}(\tilde{P}_N^1) \neq D^{\succ_A}(\tilde{P}_N^2)$, there exists a node in the OSP mechanism $\tilde{G}$ that has at least two edges. Also, note that since each agent $l \in N \setminus \{i_1, i_2, i_3\}$ has exactly one preference in $\tilde{\mathcal{P}}_l$, whenever there are at least two outgoing edges from a node, that node must be assigned to some agent in $\{i_1, i_2, i_3\}$. Consider the first node (from the root) $v$ that has at least two edges.
		\begin{enumerate}[(i)]
			\item Suppose $\eta^{NA}(v) = i_1$. 
			\\
			By OSP-implementability, the facts $a_2 \tilde{P}_{i_1}^4 a_1$, $D^{\succ_A}_{i_1}(\tilde{P}_N^4) = a_1$, and $D^{\succ_A}_{i_1}(\tilde{P}_N^1) = a_2$ together imply that $\tilde{P}_{i_1}^4$ and $\tilde{P}_{i_1}^1$ do not diverge at $v$. 
			Similarly, by OSP-implementability, the facts $a_3 \tilde{P}_{i_1}^2 a_1$, $D^{\succ_A}_{i_1}(\tilde{P}_N^2) = a_1$, and $D^{\succ_A}_{i_1}(\tilde{P}_N^3) = a_3$ together imply that $\tilde{P}_{i_1}^2$ and $\tilde{P}_{i_1}^3$ do not diverge at $v$.
			Since $\tilde{P}_{i_1}^3 = \tilde{P}_{i_1}^4$, combining all these observations, we have a contradiction to the fact that $v$ has at least two edges.
			
			\item Suppose $\eta^{NA}(v) = i_2$. 
			\\
			By OSP-implementability, the facts $a_3 \tilde{P}_{i_2}^2 a_2$, $D^{\succ_A}_{i_2}(\tilde{P}_N^2) = a_2$, and $D^{\succ_A}_{i_2}(\tilde{P}_N^1) = a_3$ together imply that $\tilde{P}_{i_2}^2$ and $\tilde{P}_{i_2}^1$ do not diverge at $v$, a contradiction to the fact that $v$ has at least two edges.	
			
			\item Suppose $\eta^{NA}(v) = i_3$. 
			\\
			By OSP-implementability, the facts $a_1 \tilde{P}_{i_3}^2 a_3$, $D^{\succ_A}_{i_3}(\tilde{P}_N^2) = a_3$, and $D^{\succ_A}_{i_3}(\tilde{P}_N^1) = a_1$ together imply that $\tilde{P}_{i_3}^2$ and $\tilde{P}_{i_3}^1$ do not diverge at $v$. 
			Similarly, by OSP-implementability, the facts $a_2 \tilde{P}_{i_3}^4 a_3$, $D^{\succ_A}_{i_3}(\tilde{P}_N^4) = a_3$, and $D^{\succ_A}_{i_3}(\tilde{P}_N^3) = a_2$ together imply that $\tilde{P}_{i_3}^4$ and $\tilde{P}_{i_3}^3$ do not diverge at $v$.
			Since $\tilde{P}_{i_3}^1 = \tilde{P}_{i_3}^3$, combining all these observations, we have a contradiction to the fact that $v$ has at least two edges.				
		\end{enumerate}
		
		\noindent\textbf{\textsc{Case} 2}: Suppose $i_2 \succ_{a_2} i_3 \succ_{a_2} i_1$ and $i_3 \succ_{a_3} i_2 \succ_{a_3} i_1$.
		
		By renaming agents $i_1, i_2, i_3$ as $i'_3, i'_2, i'_1$, respectively, and renaming objects $a_1, a_2, a_3$ as $a'_3, a'_2, a'_1$, respectively, we obtain an identical situation to Case 1.\medskip
		\\
		\noindent\textbf{\textsc{Case} 3}: Suppose $i_1 \succ_{a_1} i_3 \succ_{a_1} i_2$ and $i_3 \succ_{a_3} i_1 \succ_{a_3} i_2$.
		
		By renaming agents $i_1, i_2, i_3$ as $i'_1, i'_3, i'_2$, respectively, and renaming objects $a_1, a_2, a_3$ as $a'_1, a'_3, a'_2$, respectively, we obtain an identical situation to Case 1.\medskip
		\\
		\noindent\textbf{\textsc{Case} 4}: Suppose $i_1 \succ_{a_1} i_2 \succ_{a_1} i_3$, $i_2 \succ_{a_2} i_3 \succ_{a_2} i_1$, and $i_3 \succ_{a_3} i_1 \succ_{a_3} i_2$.
		
		Consider the domain $\tilde{\mathcal{P}}_N \subseteq \mathbb{L}^n(A \cup \{a_0\})$ with only three preference profiles presented (together with the outcomes of $D^{\succ_A}$) in Table \ref{preference choice OSP implies strong case 4}.
		\begin{table}[H]
			\centering
			\begin{tabular}{@{}c|ccccc|ccc@{}}
				\hline
				$\mbox{Preference profiles}$ & $\mbox{Agent } i_1$ & $\mbox{Agent } i_2$ & $\mbox{Agent } i_3$ & $\ldots$ & $\mbox{Agent } l$ & $D^{\succ_A}_{i_1}$ & $D^{\succ_A}_{i_2}$ & $D^{\succ_A}_{i_3}$ \\ \hline
				\hline
				$\tilde{P}_N^1$ & $a_2 a_1 a_3 a_0\ldots$ & $a_1 a_3 a_2 a_0\ldots$ & $a_2 a_1 a_3 a_0\ldots$ & $\ldots$ & $\hat{P}$ & $a_1$ & $a_3$ & $a_2$ \\ \hline
				$\tilde{P}_N^2$ & $a_3 a_2 a_1 a_0\ldots$ & $a_1 a_3 a_2 a_0\ldots$ & $a_1 a_3 a_2 a_0\ldots$ & $\ldots$ & $\hat{P}$ & $a_2$ & $a_1$ & $a_3$ \\ \hline
				$\tilde{P}_N^3$ & $a_3 a_2 a_1 a_0\ldots$ & $a_3 a_2 a_1 a_0\ldots$ & $a_2 a_1 a_3 a_0\ldots$ & $\ldots$ & $\hat{P}$ & $a_3$ & $a_2$ & $a_1$ \\ \hline
			\end{tabular}
			\caption{Preference profiles of $\tilde{\mathcal{P}}_N$ for Case 4}
			\label{preference choice OSP implies strong case 4}
		\end{table}
		
		Since $D^{\succ_A}$ is OSP-implementable on $\mathbb{L}^n(A\cup \{a_0\})$, it must be OSP-implementable on the domain $\tilde{\mathcal{P}}_N$. Let $\tilde{G}$ be an OSP mechanism that implements $D^{\succ_A}$ on $\tilde{\mathcal{P}}_N$. 
		
		Note that since $D^{\succ_A}(\tilde{P}_N^1) \neq D^{\succ_A}(\tilde{P}_N^2)$, there exists a node in the OSP mechanism $\tilde{G}$ that has at least two edges. Also, note that since each agent $l \in N \setminus \{i_1, i_2, i_3\}$ has exactly one preference in $\tilde{\mathcal{P}}_l$, whenever there are at least two outgoing edges from a node, that node must be assigned to some agent in $\{i_1, i_2, i_3\}$. Consider the first node (from the root) $v$ that has at least two edges.
		\begin{enumerate}[(i)]
			\item Suppose $\eta^{NA}(v) = i_1$. 
			\\
			By OSP-implementability, the facts $a_2 \tilde{P}_{i_1}^1 a_1$, $D^{\succ_A}_{i_1}(\tilde{P}_N^1) = a_1$, and $D^{\succ_A}_{i_1}(\tilde{P}_N^2) = a_2$ together imply that $\tilde{P}_{i_1}^1$ and $\tilde{P}_{i_1}^2$ do not diverge at $v$, a contradiction to the fact that $v$ has at least two edges.
			
			\item Suppose $\eta^{NA}(v) = i_2$. 
			\\
			By OSP-implementability, the facts $a_3 \tilde{P}_{i_2}^3 a_2$, $D^{\succ_A}_{i_2}(\tilde{P}_N^3) = a_2$, and $D^{\succ_A}_{i_2}(\tilde{P}_N^1) = a_3$ together imply that $\tilde{P}_{i_2}^3$ and $\tilde{P}_{i_2}^1$ do not diverge at $v$, a contradiction to the fact that $v$ has at least two edges.	
			
			\item Suppose $\eta^{NA}(v) = i_3$. 
			\\
			By OSP-implementability, the facts $a_1 \tilde{P}_{i_3}^2 a_3$, $D^{\succ_A}_{i_3}(\tilde{P}_N^2) = a_3$, and $D^{\succ_A}_{i_3}(\tilde{P}_N^3) = a_1$ together imply that $\tilde{P}_{i_3}^2$ and $\tilde{P}_{i_3}^3$ do not diverge at $v$, a contradiction to the fact that $v$ has at least two edges.				
		\end{enumerate}
		
		\noindent\textbf{\textsc{Case} 5}: Suppose $i_1 \succ_{a_1} i_3 \succ_{a_1} i_2$, $i_2 \succ_{a_2} i_1 \succ_{a_2} i_3$, and $i_3 \succ_{a_3} i_2 \succ_{a_3} i_1$.
		
		By renaming agents $i_1, i_2, i_3$ as $i'_1, i'_3, i'_2$, respectively, and renaming objects $a_1, a_2, a_3$ as $a'_1, a'_3, a'_2$, respectively, we obtain an identical situation to Case 4.\medskip
		
		Since Cases 1 - 5 are exhaustive, this completes the proof of Proposition \ref{proposition OSP is strong}.		
		\hfill
		\qed

		\section{Proof of Theorem \ref{theorem SOSP FPTTC restricted}}\label{appendix proof of theo SOSP FPTTC restricted}
		
		\setcounter{equation}{0}
		\setcounter{table}{0}
		\setcounter{figure}{0}

		We first make a straightforward observation to facilitate the proof.

		\begin{obs}\label{observation 2 objects}
			Suppose $|A| = 2$. On the domain $\mathbb{L}^n(A)$, every FPTTC rule is SOSP-implementable.
		\end{obs}

		\begin{proof}[\textbf{Completion of the proof of Theorem \ref{theorem SOSP FPTTC restricted}.}]
			\textbf{\textit{(If part)}} Using Observation \ref{observation 2 objects}, it is straightforward to verify that every weak serial dictatorship is SOSP-implementable on the restricted domain $\mathbb{L}^n(A)$. This completes the proof of the ``if'' part of Theorem \ref{theorem SOSP FPTTC restricted}.\medskip
			\\
			\textbf{\textit{(Only-if part)}} Let $T^{\succ_A}$ be an SOSP-implementable FPTTC rule on the restricted domain $\mathbb{L}^n(A)$. Since SOSP-implementability is stronger than OSP-implementability (see Remark \ref{remark SOSP implies OSP}), by Theorem \ref{theorem prior results}, $T^{\succ_A}$ satisfies dual ownership. Assume for contradiction that $T^{\succ_A}$ is not a weak serial dictatorship. Since $T^{\succ_A}$ satisfies dual ownership but is not a weak serial dictatorship, there exist a preference profile $P'_N \in \mathbb{L}^n(A)$ and a step $s^*$ of $T^{\succ_A}$ at $P'_N$ such that there are two agents $i, j$ and three objects $a, b, c$ in the reduced market at Step $s^*$ with the property that agent $i$ owns the object $a$, and agent $j$ owns the objects $b$ and $c$ at Step $s^*$. We distinguish the following two cases.\medskip
			\\
			\noindent\textbf{\textsc{Case} 1}: Suppose $s^* = 1$.
			
			Consider the domain $\tilde{\mathcal{P}}_N \subseteq \mathbb{L}^n(A)$ with only four preference profiles presented in Table \ref{preference choice only if case 1}. Here, $l$ denotes an agent (might be empty) other than $i$ and $j$.
			\begin{table}[H]
				\centering
				\begin{tabular}{@{}c|cccc@{}}
					\hline
					$\mbox{Preference profiles}$ & $\mbox{Agent } i$ & $\mbox{Agent } j$ & $\ldots$ & $\mbox{Agent } l$ \\ \hline
					\hline
					$\tilde{P}_N^1$ & $abc\ldots$ & $acb\ldots$ & $\ldots$ & $P'_l$ \\ \hline
					$\tilde{P}_N^2$ & $bac\ldots$ & $bac\ldots$ & $\ldots$ & $P'_l$ \\ \hline
					$\tilde{P}_N^3$ & $bca\ldots$ & $abc\ldots$ & $\ldots$ & $P'_l$ \\ \hline
					$\tilde{P}_N^4$ & $cab\ldots$ & $abc\ldots$ & $\ldots$ & $P'_l$ \\ \hline
				\end{tabular}
				\caption{Preference profiles of $\tilde{\mathcal{P}}_N$}
				\label{preference choice only if case 1}
			\end{table}
			
			In Table \ref{Outcome of hex only if part case 1}, we present some facts regarding the outcome of $T^{\succ_A}$ on the domain $\tilde{\mathcal{P}}_N$. These facts are deduced by the construction of $\tilde{\mathcal{P}}_N$ along with the assumptions for Case 1.
			\begin{table}[H]
				\centering
				\begin{tabular}{@{}c|cc|cc@{}}
					\hline
					Preference profiles & $\mbox{Agent } i$ & $\mbox{Agent } j$ & $T^{\succ_A}_i$ & $T^{\succ_A}_j$ \\ \hline
					\hline
					$\tilde{P}_N^1$ & $abc\ldots$ & $acb\ldots$ & $a$ & $c$ \\ \hline
					$\tilde{P}_N^2$ & $bac\ldots$ & $bac\ldots$ & $a$ & $b$ \\ \hline
					$\tilde{P}_N^3$ & $bca\ldots$ & $abc\ldots$ & $b$ & $a$ \\ \hline
					$\tilde{P}_N^4$ & $cab\ldots$ & $abc\ldots$ & $c$ & $a$ \\ \hline
				\end{tabular}
				\caption{Partial outcome of $T^{\succ_A}$ on $\tilde{\mathcal{P}}_N$}
				\label{Outcome of hex only if part case 1}
			\end{table}
			
			Since $T^{\succ_A}$ is SOSP-implementable on $\mathbb{L}^n(A)$, it must be SOSP-implementable on the domain $\tilde{\mathcal{P}}_N$. Let $\tilde{G}$ be an SOSP mechanism that implements $T^{\succ_A}$ on $\tilde{\mathcal{P}}_N$. 
			
			Note that since $T^{\succ_A}(\tilde{P}_N^1) \neq T^{\succ_A}(\tilde{P}_N^2)$, there exists a node in the SOSP mechanism $\tilde{G}$ that has at least two edges. Also, note that since each agent $l \in N \setminus \{i,j\}$ has exactly one preference in $\tilde{\mathcal{P}}_l$, whenever there are at least two outgoing edges from a node, that node must be assigned to some agent in $\{i,j\}$. Consider the first node (from the root) $v$ that has at least two edges. 
			\begin{enumerate}[(i)]
				\item Suppose $\eta^{NA}(v) = i$. 
				\\
				By SOSP-implementability, the facts $b \tilde{P}_i^2 a$, $T^{\succ_A}_i(\tilde{P}_N^2) = a$, and $T^{\succ_A}_i(\tilde{P}_N^3) = b$ together imply that $\tilde{P}_i^2$ and $\tilde{P}_i^3$ do not diverge at $v$. Since $\tilde{P}_i^2$ and $\tilde{P}_i^3$ do not diverge at $v$, by SOSP-implementability, the facts $c \tilde{P}_i^3 a$, $T^{\succ_A}_i(\tilde{P}_N^2) = a$, and $T^{\succ_A}_i(\tilde{P}_N^4) = c$ together imply that $\tilde{P}_i^3$ and $\tilde{P}_i^4$ do not diverge at $v$. Moreover, since $\tilde{P}_i^3$ and $\tilde{P}_i^4$ do not diverge at $v$, by SOSP-implementability, the facts $a \tilde{P}_i^4 b$, $T^{\succ_A}_i(\tilde{P}_N^1) = a$, and $T^{\succ_A}_i(\tilde{P}_N^3) = b$ together imply that $\tilde{P}_i^1$ and $\tilde{P}_i^4$ do not diverge at $v$. Combining all these observations, we have a contradiction to the fact that $v$ has at least two edges.
				
				\item Suppose $\eta^{NA}(v) = j$. 
				\\
				By SOSP-implementability, the facts $a \tilde{P}_j^1 c$, $T^{\succ_A}_j(\tilde{P}_N^1) = c$, and $T^{\succ_A}_j(\tilde{P}_N^3) = a$ together imply that $\tilde{P}_j^1$ and $\tilde{P}_j^3$ do not diverge at $v$. Since $\tilde{P}_j^1$ and $\tilde{P}_j^3$ do not diverge at $v$, by SOSP-implementability, the facts $b \tilde{P}_j^3 c$, $T^{\succ_A}_j(\tilde{P}_N^1) = c$, and $T^{\succ_A}_j(\tilde{P}_N^2) = b$ together imply that $\tilde{P}_j^2$ and $\tilde{P}_j^3$ do not diverge at $v$. Combining all these observations, we have a contradiction to the fact that $v$ has at least two edges.	
			\end{enumerate}
			
			\noindent \textbf{\textsc{Case} 2}: Suppose $s^* > 1$.
			
			Recall that $X^{s^*-1}(P'_N)$ is the set of assigned objects up to Stage $s^* - 1$ (including Stage $s^* - 1$) of $T^{\succ_A}$ at $P'_N$. Fix a preference $\hat{P} \in \mathbb{L}(X^{s^*-1}(P'_N))$ over these objects. Consider the domain $\tilde{\mathcal{P}}_N \subseteq \mathbb{L}^n(A)$ with only four preference profiles presented in Table \ref{preference choice only if case 2}.\footnote{For instance, $\hat{P} a b c \ldots$ denotes a preference where objects in $X^{s^*-1}(P'_N)$ are ranked at the top according to the preference $\hat{P}$, objects $a$, $b$, and $c$ are ranked consecutively after that (in that order), and the ranking of the rest of the objects is arbitrarily.}
			\begin{table}[H]
				\centering
				\begin{tabular}{@{}c|cccc@{}}
					\hline
					$\mbox{Preference profiles}$ & $\mbox{Agent } i$ & $\mbox{Agent } j$ & $\ldots$ & $\mbox{Agent } l$ \\ \hline
					\hline
					$\tilde{P}_N^1$ & $\hat{P}abc\ldots$ & $\hat{P}acb\ldots$ & $\ldots$ & $P'_l$ \\ \hline
					$\tilde{P}_N^2$ & $\hat{P}bac\ldots$ & $\hat{P}bac\ldots$ & $\ldots$ & $P'_l$ \\ \hline
					$\tilde{P}_N^3$ & $\hat{P}bca\ldots$ & $\hat{P}abc\ldots$ & $\ldots$ & $P'_l$ \\ \hline
					$\tilde{P}_N^4$ & $\hat{P}cab\ldots$ & $\hat{P}abc\ldots$ & $\ldots$ & $P'_l$ \\ \hline
				\end{tabular}
				\caption{Preference profiles of $\tilde{\mathcal{P}}_N$}
				\label{preference choice only if case 2}
			\end{table}	
			
			In Table \ref{Outcome of hex only if part case 2}, we present some facts regarding the outcome of $T^{\succ_A}$ on the domain $\tilde{\mathcal{P}}_N$ that can be deduced by the construction of the domain $\tilde{\mathcal{P}}_N$ along with the assumptions for Case 2. The verification of these facts is left to the reader.
			\begin{table}[H]
				\centering
				\begin{tabular}{@{}c|cc|cc@{}}
					\hline
					Preference profiles & $\mbox{Agent } i$ & $\mbox{Agent } j$ & $T^{\succ_A}_i$ & $T^{\succ_A}_j$ \\ \hline
					\hline
					$\tilde{P}_N^1$ & $\hat{P}abc\ldots$ & $\hat{P}acb\ldots$ & $a$ & $c$ \\ \hline
					$\tilde{P}_N^2$ & $\hat{P}bac\ldots$ & $\hat{P}bac\ldots$ & $a$ & $b$ \\ \hline
					$\tilde{P}_N^3$ & $\hat{P}bca\ldots$ & $\hat{P}abc\ldots$ & $b$ & $a$ \\ \hline
					$\tilde{P}_N^4$ & $\hat{P}cab\ldots$ & $\hat{P}abc\ldots$ & $c$ & $a$ \\ \hline
				\end{tabular}
				\caption{Partial outcome of $T^{\succ_A}$ on $\tilde{\mathcal{P}}_N$}
				\label{Outcome of hex only if part case 2}
			\end{table}
			
			Using a similar argument as for Case 1, we get a contradiction. This completes the proof of the ``only-if'' part of Theorem \ref{theorem SOSP FPTTC restricted}.		
		\end{proof}

		\section{Proof of Theorem \ref{theorem equivalent condition to bipolar seq}}\label{appendix proof of theo equivalent condition to bipolar seq}
		
		\setcounter{equation}{0}
		\setcounter{table}{0}
		\setcounter{figure}{0}

		The ``if'' part of the theorem is straightforward. We proceed to prove the ``only-if'' part. To do so, we prove the contrapositive. Suppose there exist an agent $i^* \in N$ and two objects $a^*, b^* \in A$ such that $rank(i^*, \succ_{a^*}) \leq |A| - 2$ and $rank(i^*, \succ_{a^*}) \neq rank(i^*, \succ_{b^*})$. Without loss of generality, assume that for all $l \in U(i^*, \succ_{a^*})$, $rank(l, \succ_{a^*}) = rank(l, \succ_{b})$ for all $b \in A$. Let $rank(i^*, \succ_{a^*}) = m^*$ and let $A' \subseteq A \setminus \{a^*, b^*\}$ be such that $|A'| = m^* - 1$. Clearly, $m^* \leq |A| - 2$. Furthermore, $A'$ is well-defined since $m^* \leq |A| - 2$.
		
		Fix a preference $\hat{P} \in \mathbb{L}(A)$ such that $A' \hat{P} (A \setminus A')$. Consider the preference profile $\tilde{P}_N \in \mathbb{L}^n(A)$ such that $\tilde{P}_i = \hat{P}$ for all $i \in N$. Since $rank(l, \succ_{a^*}) = rank(l, \succ_{b})$ for all $l \in U(i^*, \succ_{a^*})$ and all $b \in A$, and $rank(i^*, \succ_{a^*}) = m^*$, it follows from the construction of $\tilde{P}_N$ that $I^{m^*-1}(\tilde{P}_N) = U(i^*, \succ_{a^*})$ and $X^{m^*-1}(\tilde{P}_N) = A'$. The facts $A' \subseteq A \setminus \{a^*, b^*\}$ and $X^{m^*-1}(\tilde{P}_N) = A'$ together imply $a^*, b^* \in A_{m^*}(\tilde{P}_N)$. Since $I^{m^*-1}(\tilde{P}_N) = U(i^*, \succ_{a^*})$, $rank(l, \succ_{a^*}) = rank(l, \succ_{b^*})$ for all $l \in U(i^*, \succ_{a^*})$, $rank(i^*, \succ_{a^*}) \neq rank(i^*, \succ_{b^*})$, and $a^*, b^* \in A_{m^*}(\tilde{P}_N)$, it follows that $|N_{m^*}(\tilde{P}_N)| \geq 2$. Moreover, since $X^{m^*-1}(\tilde{P}_N) = A'$, $|A'| = m^* - 1$, and $m^* \leq |A| - 2$, we have $|A_{m^*}(\tilde{P}_N)| \geq 3$. However, the facts $|N_{m^*}(\tilde{P}_N)| \geq 2$ and $|A_{m^*}(\tilde{P}_N)| \geq 3$ together imply that $T^{\succ_A}$ is not a weak serial dictatorship. This completes the proof of the ``only-if'' part of Theorem \ref{theorem equivalent condition to bipolar seq}.	
		\hfill
		\qed

		\section{Proof of Theorem \ref{theorem SOSP FPTTC unrestricted}}\label{appendix proof of theo SOSP FPTTC unrestricted}
		
		\setcounter{equation}{0}
		\setcounter{table}{0}
		\setcounter{figure}{0}

		The ``if'' part of the theorem is straightforward. We proceed to prove the ``only-if'' part. Let $T^{\succ_A}$ be an SOSP-implementable FPTTC rule on the unrestricted domain $\mathbb{L}^n(A \cup \{a_0\})$. Assume for contradiction that $T^{\succ_A}$ is not a serial dictatorship. Then, there exist two agents $i,j \in N$ and two objects $a,b \in A$ such that $i \succ_a j$ and $j \succ_b i$.
		
		Fix a preference $\hat{P} \in \mathbb{L}(A \cup \{a_0\})$ such that $\tau(\hat{P}) = a_0$. Consider the domain $\tilde{\mathcal{P}}_N \subseteq \mathbb{L}^n(A \cup \{a_0\})$ with only three preference profiles presented (together with the outcomes of $T^{\succ_A}$) in Table \ref{preference choice only if unrestricted}. Here, $l$ denotes an agent (might be empty) other than $i$ and $j$.
		\begin{table}[H]
			\centering
			\begin{tabular}{@{}c|cccc|cc@{}}
				\hline
				$\mbox{Preference profiles}$ & $\mbox{Agent } i$ & $\mbox{Agent } j$ & $\ldots$ & $\mbox{Agent } l$ & $T^{\succ_A}_i$ & $T^{\succ_A}_j$ \\ \hline
				\hline
				$\tilde{P}_N^1$ & $a b a_0\ldots$ & $a a_0\ldots$ & $\ldots$ & $\hat{P}$ & $a$ & $a_0$ \\ \hline
				$\tilde{P}_N^2$ & $b a_0\ldots$ & $b a a_0\ldots$ & $\ldots$ & $\hat{P}$ & $a_0$ & $b$ \\ \hline
				$\tilde{P}_N^3$ & $b a a_0\ldots$ & $a b a_0\ldots$ & $\ldots$ & $\hat{P}$ & $b$ & $a$ \\ \hline
			\end{tabular}
			\caption{Preference profiles of $\tilde{\mathcal{P}}_N$}
			\label{preference choice only if unrestricted}
		\end{table}
		
		Since $T^{\succ_A}$ is SOSP-implementable on $\mathbb{L}^n(A\cup \{a_0\})$, it must be SOSP-implementable on the domain $\tilde{\mathcal{P}}_N$. Let $\tilde{G}$ be an SOSP mechanism that implements $T^{\succ_A}$ on $\tilde{\mathcal{P}}_N$. 
		
		Note that since $T^{\succ_A}(\tilde{P}_N^1) \neq T^{\succ_A}(\tilde{P}_N^2)$, there exists a node in the SOSP mechanism $\tilde{G}$ that has at least two edges. Also, note that since each agent $l \in N \setminus \{i,j\}$ has exactly one preference in $\tilde{\mathcal{P}}_l$, whenever there are at least two outgoing edges from a node, that node must be assigned to some agent in $\{i,j\}$. Consider the first node (from the root) $v$ that has at least two edges. We distinguish the following two cases.\medskip
		\\
		\noindent\textbf{\textsc{Case} 1}: Suppose $\eta^{NA}(v) = i$.
		
		By SOSP-implementability, the facts $b \tilde{P}_i^2 a_0$, $T^{\succ_A}_i(\tilde{P}_N^2) = a_0$, and $T^{\succ_A}_i(\tilde{P}_N^3) = b$ together imply that $\tilde{P}_i^2$ and $\tilde{P}_i^3$ do not diverge at $v$. Since $\tilde{P}_i^2$ and $\tilde{P}_i^3$ do not diverge at $v$, by SOSP-implementability, the facts $a \tilde{P}_i^3 a_0$, $T^{\succ_A}_i(\tilde{P}_N^2) = a_0$, and $T^{\succ_A}_i(\tilde{P}_N^1) = a$ together imply that $\tilde{P}_i^1$ and $\tilde{P}_i^3$ do not diverge at $v$. Combining all these observations, we have a contradiction to the fact that $v$ has at least two edges.\medskip
		\\
		\noindent\textbf{\textsc{Case} 2}: Suppose $\eta^{NA}(v) = j$.
		
		By SOSP-implementability, the facts $a \tilde{P}_j^1 a_0$, $T^{\succ_A}_j(\tilde{P}_N^1) = a_0$, and $T^{\succ_A}_j(\tilde{P}_N^3) = a$ together imply that $\tilde{P}_j^1$ and $\tilde{P}_j^3$ do not diverge at $v$. Since $\tilde{P}_j^1$ and $\tilde{P}_j^3$ do not diverge at $v$, by SOSP-implementability, the facts $b \tilde{P}_j^3 a_0$, $T^{\succ_A}_j(\tilde{P}_N^1) = a_0$, and $T^{\succ_A}_j(\tilde{P}_N^2) = b$ together imply that $\tilde{P}_j^2$ and $\tilde{P}_j^3$ do not diverge at $v$. Combining all these observations, we have a contradiction to the fact that $v$ has at least two edges.\medskip
		
		Since Cases 1 and 2	are exhaustive, this completes the proof of the ``only-if'' part of Theorem \ref{theorem SOSP FPTTC unrestricted}.		
		\hfill
		\qed

		\section{Proofs of Proposition \ref{proposition SSP implies SOSP}}\label{appendix proof of proposition SSP implies SOSP}
		
		\setcounter{equation}{0}
		\setcounter{table}{0}
		\setcounter{figure}{0}

		Fix an arbitrary domain of preference profiles $\mathcal{P}_N$. Let $G$ be a simple OSP mechanism on $\mathcal{P}_N$. Consider the assignment rule $f^G$ on $\mathcal{P}_N$ implemented by $G$. Consider an agent $i \in N$, a node $v$ such that $\eta^{NA}(v) = i$, and preference profiles $P_N, P'_N, \tilde{P}_N \in \mathcal{P}_N$ passing through $v$ such that (i) $P_i$ and $P'_i$ do not diverge at $v$ and (ii) $P_i$ and $\tilde{P}_i$ diverge at $v$. We show that $f^G_i(P'_N) R_i f^G_i(\tilde{P}_N)$.
		
		Since $P_i$ and $P'_i$ do not diverge at $v$, the fact that $G$ is a simple mechanism implies that $f^G(P'_N) = f^G(P_i, P'_{-i})$. This, in particular, means 
		\begin{equation}\label{equation simple implies SOSP}
			f^G_i(P'_N) = f^G_i(P_i, P'_{-i}).
		\end{equation}
		The fact that both $P_N$ and $P'_N$ pass through $v$ implies that $(P_i, P'_{-i})$ passes through $v$. Consider the preference profiles $(P_i, P'_{-i})$ and $\tilde{P}_N$. Since both of them pass through $v$ at which $P_i$ and $\tilde{P}_i$ diverge, by obvious strategy-proofness of $G$, we have $f^G_i(P_i, P'_{-i}) R_i f^G_i(\tilde{P}_N)$. This, together with \eqref{equation simple implies SOSP}, implies $f^G_i(P'_N) R_i f^G_i(\tilde{P}_N)$. This completes the proof of Proposition \ref{proposition SSP implies SOSP}.
		\hfill
		\qed

		\section{Proof of Theorem \ref{theorem SSP SOSP FPTTC}}\label{appendix proof of theorem SSP SOSP FPTTC}
		
		\setcounter{equation}{0}
		\setcounter{table}{0}
		\setcounter{figure}{0}

		We first make a straightforward observation to facilitate the proof.

		\begin{obs}\label{observation 2 objects SSP}
			Suppose $|A| = 2$. On the domain $\mathbb{L}^n(A)$, every FPTTC rule is simply strategy-proof.
		\end{obs}

		\begin{proof}[\textbf{Completion of the proof of Theorem \ref{theorem SSP SOSP FPTTC}.}]
			The ``only-if'' part of the theorem follows from Proposition \ref{proposition SSP implies SOSP}, we proceed
			to prove the ``if'' part. Let $T^{\succ_A}$ be an SOSP-implementable FPTTC rule on $\mathcal{P}_N$. We distinguish the following two cases.\medskip
			\\
			\noindent\textbf{\textsc{Case} 1}: Suppose $\mathcal{P}_N = \mathbb{L}^n(A)$.
			
			From Theorem \ref{theorem SOSP FPTTC restricted}, it follows that $T^{\succ_A}$ is a weak serial dictatorship. This, along with Observation \ref{observation 2 objects SSP}, implies that $T^{\succ_A}$ is simply strategy-proof.\medskip
			\\
			\noindent\textbf{\textsc{Case} 2}: Suppose $\mathcal{P}_N = \mathbb{L}^n(A \cup \{a_0\})$.
			
			The proof for this case is straightforward.
		\end{proof}

		\section{Relation between SOSP-implementability and simple strategy-proofness beyond FPTTC rules}\label{appendix examples}
		
		\setcounter{equation}{0}
		\setcounter{table}{0}
		\setcounter{figure}{0}

		Example \ref{example sosp but not ssp restricted} shows that not every SOSP-implementable assignment rule is simply strategy-proof on the restricted domain $\mathbb{L}^n(A)$. A similar result for the unrestricted domain $\mathbb{L}^n(A \cup \{a_0\})$ is shown in Example \ref{example sosp but not ssp unrestricted}. It should be noted that the assignment rules constructed in these two examples are not Pareto efficient.

		\begin{example}\label{example sosp but not ssp restricted}
			Consider an allocation problem with three agents $N = \{1, 2, 3\}$ and three objects $A = \{a_1, a_2, a_3\}$. In Figure \ref{tree sosp not ssp restricted}, we provide an SOSP mechanism $G$ on $\mathbb{L}^3(A)$.

			\begin{figure}[H]
				\centering
				\begin{tikzpicture}
					[
					grow                    = down,
					sibling distance        = 18em,
					level distance          = 5em,
					edge from parent/.style = {draw, -latex},
					every node/.style       = {font=\footnotesize}
					]
					\node [root] {1}
					child { node [root] {2}
						[
						grow                    = down,
						sibling distance        = 6em,
						level distance          = 5em,
						edge from parent/.style = {draw, -latex},
						every node/.style       = {font=\footnotesize}
						]
						child { node [env] {$(a_1,a_2,a_3)$}
							edge from parent node [left] {$\cdot a_2 \cdot a_3 \cdot$} }
						child { node [env] {$(a_1,a_3,a_2)$}
							edge from parent node [right] {$\cdot a_3 \cdot a_2 \cdot$} }
						edge from parent node [left] {$\cdot a_1 \cdot a_2 \cdot$} }
					child { node [root] {2}
						[
						grow                    = down,
						sibling distance        = 6em,
						level distance          = 5em,
						edge from parent/.style = {draw, -latex},
						every node/.style       = {font=\footnotesize}
						]
						child { node [env] {$(a_2,a_1,a_3)$}
							edge from parent node [left] {$\tau(P_2) = a_1$} }
						child { node [env] {$(a_1,a_2,a_3)$}
							edge from parent node {$\tau(P_2) = a_2$} }
						child { node [root] {1}
							[
							grow                    = down,
							sibling distance        = 5em,
							level distance          = 5em,
							edge from parent/.style = {draw, -latex},
							every node/.style       = {font=\footnotesize}
							]
							child { node [env] {$(a_2, a_3, a_1)$}
								edge from parent node [left] {$\tau(P_1) = a_2$} }
							child { node [env] {$(a_2, a_3, a_0)$}
								edge from parent node [right] {$a_3 a_2 a_1 a_0$} }
							edge from parent node [right] {$\tau(P_2) = a_3$} }
						edge from parent node [right] {$\cdot a_2 \cdot a_1 \cdot$} };
				\end{tikzpicture}
				\caption{SOSP mechanism $G$ for Example \ref{example sosp but not ssp restricted}}
				\label{tree sosp not ssp restricted}
			\end{figure}

			Consider the assignment rule $f^G$ on $\mathbb{L}^3(A)$ implemented by $G$. By definition, $f^G$ is SOSP-implementable on $\mathbb{L}^3(A)$. We argue that $f^G$ is not simply strategy-proof on $\mathbb{L}^3(A)$. To do that, it is enough to show that $f^G$ is not simply strategy-proof on some subdomain of $\mathbb{L}^3(A)$.
			
			Consider the domain $\tilde{\mathcal{P}}_N \subseteq \mathbb{L}^3(A)$ with only five preference profiles presented (together with the outcomes of $f^G$) in Table \ref{preference choice not SSP restricted}.
			\begin{table}[H]
				\centering
				\begin{tabular}{@{}c|ccc|ccc@{}}
					\hline
					$\mbox{Preference profiles}$ & $\mbox{Agent } 1$ & $\mbox{Agent } 2$ & $\mbox{Agent } 3$ & $f^G_1$ & $f^G_2$ & $f^G_3$ \\ \hline
					\hline
					$\tilde{P}_N^1$ & $a_1 a_2 a_3 a_0$ & $a_1 a_2 a_3 a_0$ & $a_1 a_2 a_3 a_0$ & $a_1$ & $a_2$ & $a_3$ \\ \hline
					$\tilde{P}_N^2$ & $a_2 a_1 a_3 a_0$ & $a_1 a_3 a_2 a_0$ & $a_1 a_2 a_3 a_0$ & $a_2$ & $a_1$ & $a_3$ \\ \hline
					$\tilde{P}_N^3$ & $a_2 a_1 a_3 a_0$ & $a_3 a_1 a_2 a_0$ & $a_1 a_2 a_3 a_0$ & $a_2$ & $a_3$ & $a_1$ \\ \hline
					$\tilde{P}_N^4$ & $a_3 a_2 a_1 a_0$ & $a_2 a_1 a_3 a_0$ & $a_1 a_2 a_3 a_0$ & $a_1$ & $a_2$ & $a_3$ \\ \hline
					$\tilde{P}_N^5$ & $a_3 a_2 a_1 a_0$ & $a_3 a_1 a_2 a_0$ & $a_1 a_2 a_3 a_0$ & $a_2$ & $a_3$ & $a_0$ \\ \hline
				\end{tabular}
				\caption{Preference profiles for Example \ref{example sosp but not ssp restricted}}
				\label{preference choice not SSP restricted}
			\end{table}
			
			Assume for contradiction that $f^G$ is simply strategy-proof on $\tilde{\mathcal{P}}_N$. So, there exists a simple OSP mechanism $\tilde{G}$ that implements $f^G$ on $\tilde{\mathcal{P}}_N$. Note that since $f^G(\tilde{P}_N^1) \neq f^G(\tilde{P}_N^2)$, there exists a node in the simple OSP mechanism $\tilde{G}$ that has at least two edges. Also, note that since agent 3 has exactly one preference in $\tilde{\mathcal{P}}_3$, whenever there are at least two outgoing edges from a node, the node must be assigned either to agent 1 or to agent 2. Consider the first node (from the root) $v$ that has at least two edges. We distinguish the following two cases.
			\begin{enumerate}[(i)]
				\item\label{item SOSP not SSP restricted 1} Suppose $\eta^{NA}(v) = 1$. 
				\\
				By obvious strategy-proofness of $\tilde{G}$, the facts $a_2 \tilde{P}_1^4 a_1$, $f^G_1(\tilde{P}_N^4) = a_1$, and $f^G_1(\tilde{P}_N^3) = a_2$ together imply that $\tilde{P}_1^3$ and $\tilde{P}_1^4$ do not diverge at $v$. This, together with the facts that $\tilde{P}_1^4 = \tilde{P}_1^5$, $\tilde{P}_2^3 = \tilde{P}_2^5$, and $f^G(\tilde{P}_N^3) \neq f^G(\tilde{P}_N^5)$, implies that there exists a node $v'$ at which $\tilde{P}_1^3$ and $\tilde{P}_1^4$ diverge. Clearly, $v$ and $v'$ are distinct nodes appearing in the same path such that $\eta^{NA}(v) = \eta^{NA}(v') = 1$. This contradicts the fact that $\tilde{G}$ is a simple mechanism.
				
				\item\label{item SOSP not SSP restricted 2} Suppose $\eta^{NA}(v) = 2$. 
				\\
				First note that since $\tilde{G}$ is a simple OSP mechanism on $\tilde{\mathcal{P}}_N$, by Proposition \ref{proposition SSP implies SOSP}, $\tilde{G}$ is SOSP on $\tilde{\mathcal{P}}_N$. By obvious strategy-proofness of $\tilde{G}$, the facts $a_1 \tilde{P}_2^1 a_2$, $f^G_2(\tilde{P}_N^1) = a_2$, and $f^G_2(\tilde{P}_N^2) = a_1$ together imply that $\tilde{P}_2^1$ and $\tilde{P}_2^2$ do not diverge at $v$. Since $\tilde{P}_2^1$ and $\tilde{P}_2^2$ do not diverge at $v$, by strongly obvious strategy-proofness of $\tilde{G}$, the facts $a_3 \tilde{P}_2^2 a_2$, $f^G_2(\tilde{P}_N^1) = a_2$, and $f^G_2(\tilde{P}_N^3) = a_3$ together imply that $\tilde{P}_2^2$ and $\tilde{P}_2^3$ do not diverge at $v$. This, together with the facts that $\tilde{P}_1^2 = \tilde{P}_1^3$ and $f^G(\tilde{P}_N^2) \neq f^G(\tilde{P}_N^3)$, implies that there exists a node $v'$ at which $\tilde{P}_2^2$ and $\tilde{P}_2^3$ diverge. Clearly, $v$ and $v'$ are distinct nodes appearing in the same path such that $\eta^{NA}(v) = \eta^{NA}(v') = 2$. This contradicts the fact that $\tilde{G}$ is a simple mechanism.
			\end{enumerate}
			Since Cases \ref{item SOSP not SSP restricted 1} and \ref{item SOSP not SSP restricted 2} are exhaustive, it follows that $f^G$ is not simply strategy-proof on $\tilde{\mathcal{P}}_N$.
			\hfill
			$\Diamond$
		\end{example}

		\begin{example}\label{example sosp but not ssp unrestricted}
			Consider an allocation problem with three agents $N = \{1, 2, 3\}$ and three objects $A = \{a_1, a_2, a_3\}$. In Figure \ref{tree sosp not ssp unrestricted}, we provide an SOSP mechanism $G$ on $\mathbb{L}^3(A \cup \{a_0\})$. We use the following notation in Figure \ref{tree sosp not ssp unrestricted}: by $\cdot a_1 \cdot \{a_2, a_3\} \cdot$, we denote the set of preferences where $a_1$ is preferred to both $a_2$ and $a_3$.

			\begin{figure}[H]
				\centering
				\begin{tikzpicture}
					[
					grow                    = down,
					sibling distance        = 18em,
					level distance          = 5em,
					edge from parent/.style = {draw, -latex},
					every node/.style       = {font=\footnotesize}
					]
					\node [root] {1}
					child { node [root] {2}
						[
						grow                    = down,
						sibling distance        = 6em,
						level distance          = 5em,
						edge from parent/.style = {draw, -latex},
						every node/.style       = {font=\footnotesize}
						]
						child { node [env] {$(a_0,a_0,a_0)$}
							edge from parent node [left] {$\cdot a_0 \cdot a_2 \cdot$} }
						child { node [env] {$(a_0,a_2,a_0)$}
							edge from parent node [right] {$\cdot a_2 \cdot a_0 \cdot$} }
						edge from parent node [left] {$\cdot a_0 \cdot a_1 \cdot$} }
					child { node [root] {2}
						[
						grow                    = down,
						sibling distance        = 6em,
						level distance          = 5em,
						edge from parent/.style = {draw, -latex},
						every node/.style       = {font=\footnotesize}
						]
						child { node [env] {$(a_1,a_0,a_0)$}
							edge from parent node [left] {$\cdot a_0 \cdot \{a_1, a_2\} \cdot$} }
						child { node [env] {$(a_0,a_1,a_0)$}
							edge from parent node {$\cdot a_1 \cdot \{a_0, a_2\} \cdot$} }
						child { node [root] {3}
							[
							grow                    = down,
							sibling distance        = 5em,
							level distance          = 5em,
							edge from parent/.style = {draw, -latex},
							every node/.style       = {font=\footnotesize}
							]
							child { node [root] {1}
								child { node [env] {$(a_1, a_2, a_3)$}
									edge from parent node [left] {$\cdot a_3 \cdot a_0 \cdot$} }
								child { node [env] {$(a_1, a_2, a_0)$}
									edge from parent node [right] {$\cdot a_0 \cdot a_3 \cdot$} }
								edge from parent node [left] {$\cdot a_3 \cdot a_0 \cdot$} }
							child { node [env] {$(a_1, a_2, a_0)$}
								edge from parent node [right] {$\cdot a_0 \cdot a_3 \cdot$} }
							edge from parent node [right] {$\cdot a_2 \cdot \{a_0, a_1\} \cdot$} }
						edge from parent node [right] {$\cdot a_1 \cdot a_0 \cdot$} };
				\end{tikzpicture}
				\caption{SOSP mechanism $G$ for Example \ref{example sosp but not ssp unrestricted}}
				\label{tree sosp not ssp unrestricted}
			\end{figure}

			Consider the assignment rule $f^G$ on $\mathbb{L}^3(A \cup \{a_0\})$ implemented by $G$. By definition, $f^G$ is SOSP-implementable on $\mathbb{L}^3(A \cup \{a_0\})$. We argue that $f^G$ is not simply strategy-proof on $\mathbb{L}^3(A \cup \{a_0\})$. To do that, it is enough to show that $f^G$ is not simply strategy-proof on some subdomain of $\mathbb{L}^3(A \cup \{a_0\})$.
			
			Consider the domain $\tilde{\mathcal{P}}_N \subseteq \mathbb{L}^3(A \cup \{a_0\})$ with only five preference profiles presented (together with the outcomes of $f^G$) in Table \ref{preference choice not SSP unrestricted}.
			\begin{table}[H]
				\centering
				\begin{tabular}{@{}c|ccc|ccc@{}}
					\hline
					$\mbox{Preference profiles}$ & $\mbox{Agent } 1$ & $\mbox{Agent } 2$ & $\mbox{Agent } 3$ & $f^G_1$ & $f^G_2$ & $f^G_3$ \\ \hline
					\hline
					$\tilde{P}_N^1$ & $a_0 a_1 a_2 a_3$ & $a_1 a_0 a_2 a_3$ & $a_3 a_2 a_1 a_0$ & $a_0$ & $a_0$ & $a_0$ \\ \hline
					$\tilde{P}_N^2$ & $a_1 a_3 a_0 a_2$ & $a_1 a_2 a_0 a_3$ & $a_3 a_2 a_1 a_0$ & $a_0$ & $a_1$ & $a_0$ \\ \hline
					$\tilde{P}_N^3$ & $a_1 a_3 a_0 a_2$ & $a_2 a_0 a_1 a_3$ & $a_3 a_2 a_1 a_0$ & $a_1$ & $a_2$ & $a_3$ \\ \hline
					$\tilde{P}_N^4$ & $a_2 a_1 a_0 a_3$ & $a_1 a_0 a_2 a_3$ & $a_3 a_2 a_1 a_0$ & $a_0$ & $a_1$ & $a_0$ \\ \hline
					$\tilde{P}_N^5$ & $a_2 a_1 a_0 a_3$ & $a_2 a_0 a_1 a_3$ & $a_3 a_2 a_1 a_0$ & $a_1$ & $a_2$ & $a_0$ \\ \hline
				\end{tabular}
				\caption{Preference profiles for Example \ref{example sosp but not ssp unrestricted}}
				\label{preference choice not SSP unrestricted}
			\end{table}
			
			Using a similar argument as for Example \ref{example sosp but not ssp restricted}, it follows from Table \ref{preference choice not SSP unrestricted} that $f^G$ is not simply strategy-proof on $\tilde{\mathcal{P}}_N$.
			\hfill
			$\Diamond$
		\end{example}

	\end{appendices}

	\setcitestyle{numbers}
	\bibliographystyle{plainnat}
	\bibliography{mybib}

\end{document}